\documentclass[journal]{IEEEtran}
\usepackage{xcolor,soul,framed} 
\colorlet{shadecolor}{yellow}
\usepackage[pdftex]{graphicx}
\graphicspath{{../pdf/}{../jpeg/}}
\DeclareGraphicsExtensions{.pdf,.jpeg,.png}
\usepackage[cmex10]{amsmath}
\usepackage{array}
\usepackage{subcaption}
\usepackage{amssymb}
\usepackage{mdwmath}

\usepackage{mdwtab}
\usepackage{eqparbox}
\usepackage{url}
\usepackage{enumerate}
\usepackage{enumitem}
\usepackage{amsfonts}
\usepackage[english]{babel}
\usepackage[utf8]{inputenc}
\usepackage{bm}
\usepackage[linesnumbered,ruled]{algorithm2e}
\usepackage{tabularx}
\usepackage{nicematrix}
\usepackage{blkarray}

\usepackage[noend]{algpseudocode}
\usepackage{mathtools}
\usepackage{nomencl}
\usepackage{blindtext}
\usepackage{amsmath}
\usepackage[colorinlistoftodos]{todonotes}
\usepackage{color}
\usepackage{color,soul}
\usepackage{xcolor}
\usepackage{multirow}
\usepackage{graphicx}
\usepackage{epstopdf}
\usepackage[T1]{fontenc}
\usepackage{textcase}
\usepackage{array}
\usepackage{mdwmath}
\usepackage{mdwtab}
\usepackage{eqparbox}
\usepackage{enumerate}
\usepackage{amsfonts}
\usepackage[ruled]{algorithm2e}
\usepackage{booktabs}
\usepackage{comment}
\usepackage{amsmath}
\usepackage{mathtools}
\usepackage{amsthm}
\newtheorem{remark}{Remark}
\usepackage{cite}
\usepackage{url}
\usepackage{setspace}
\newtheorem{definition}{Definition}
\newtheorem{proposition}{Proposition}

\newtheorem{lemma}{Lemma}

\usepackage{etoolbox}
\makeatletter
\patchcmd{\@makecaption}
  {\scshape}
  {}
  {}
  {} 
\makeatother
\usepackage{hyperref}

\title{Robust Moving Target Defence Against False Data Injection Attacks in Power Grids}

\author{Wangkun Xu, \IEEEmembership{Student Member,~IEEE},
        Imad M. Jaimoukh, and
        Fei Teng, \IEEEmembership{Senior Member,~IEEE}
        
\thanks{This work was supported by EPSRC under Grant EP/W028662/1 and by The Royal Society under Grant RGS/R1/211256. (\textit{Corresponding author: Fei Teng})

The authors are with the Department of Electrical and Electronic Engineering, Imperial College London, London, SW7 2AZ, U.K. }} 

\begin{document}

\renewcommand{\baselinestretch}{1}
\newcommand{\squeezeup}{\vspace{0mm}}
\markboth{This paper has been accepted by IEEE Trans. on information forensics and security. Copyright of the paper is reserved by IEEE.}%
{Shell \MakeLowercase{\textit{et al.}}: Bare Demo of IEEEtran.cls for IEEE Journals}

\maketitle

\begin{abstract}
Recently, moving target defence (MTD) has been proposed to thwart false data injection (FDI) attacks in power system state estimation by proactively triggering the distributed flexible AC transmission system (D-FACTS) devices. One of the key challenges for MTD in power grid is to design its real-time implementation with performance guarantees against unknown attacks. Converting from the noiseless assumptions in the literature, this paper investigates the MTD design problem in a noisy environment and proposes, for the first time, the concept of robust MTD to guarantee the worst-case detection rate against all unknown attacks. We theoretically prove that, for any given MTD strategy, the minimal principal angle between the Jacobian subspaces corresponds to the worst-case performance against all potential attacks. Based on this finding, robust MTD algorithms are formulated for the systems with both complete and incomplete configurations. Extensive simulations using standard IEEE benchmark systems demonstrate the improved average and worst-case performances of the proposed robust MTD against state-of-the-art algorithms. All codes are available at \url{https://github.com/xuwkk/Robust_MTD}.

\end{abstract}


\IEEEpeerreviewmaketitle

\section{Introduction}

\subsection{Background}

\IEEEPARstart{T}{HE EMERGING} implementation of information techniques has reformed the power gird into a complex cyber-physical power system (CPPS), where the two-way real-time communication among multiple parties raises new risks in the grid
\cite{ten2008vulnerability}. Musleh \textit{et al.} \cite{musleh2019survey} reviewed seven recent cyber attacks in energy industry and spotted the related vulnerabilities in both physical and cyber layers. Recently, false data injection (FDI) attacks against power system state estimation (SE) have been developed by intruding through the Modbus/TCP protocol without being noticed by the bad data detector (BDD) at the control centre \cite{liu2011false,hug2012vulnerability, rahman2013false, ahmed2019unsupervised}. As accurate state estimation is crucial for energy management system (EMS) activities, such as generator dispatch, contingency analysis, and fault diagnosis, states falsified by FDI attacks can result in erroneous control actions, causing economic losses, system instability, and safety violation \cite{gomez2018electric, tajer2017false, xie2011integrity}. 

As the power system operates quasi-statically, the intruders have enough time to learn the system parameters and prepare FDI attacks \cite{kim2014subspace, yu2015blind,lakshminarayana2021data}. As a result, it is crucial to invalidate the attacker's knowledge by proactively changing the system configuration. Moving target defence (MTD), which is conceptualised first for information technology security, utilises this proactive defence idea \cite{cho2020toward}. With the distributed flexible AC transmission system (D-FACTS) devices, the control centre can alter the reactances of the transmission lines to physically change the system parameters that 
are unknown to the attackers.

\subsection{Related Work}

Initially, MTD research involves using random placement and reactance perturbations to expose FDI attacks \cite{morrow2012topology, davis2012power, rahman2014moving}. However, it has been shown that the so-called `naive' applications cannot guarantee an effective detection on stealthy FDI attacks. Therefore, \cite{liu2018reactance} and \cite{zhang2019analysis} demonstrate that the effectiveness of MTD depends on the rank of the composite pre- and post- MTD measurement matrices. Furthermore, Liu, \textit{et al.} \cite{liu2020optimal} and Zhang \textit{et al} \cite{zhang2021strategic} investigate the D-FACTS devices placement in the planning stage to maximise the effectiveness while minimising the investment budget. The authors in \cite{lakshminarayana2021cost} analyse the effectiveness of the MTD using the minimal principal angle metric and numerically show the relationship between the angle and the average detection rate, which can be used to design the MTD. Liu, \textit{et al.} \cite{liu2019joint} extends the MTD strategy in \cite{liu2018reactance} with sensor protections and Tian, \textit{et al.} \cite{tian2019moving} applies MTD to detect Stuxnet-like attack. Moreover, Higgins \textit{et.al.} \cite{higgins2020stealthy} suggests to perturb the reactance through Gaussian watermarking to prevent the attacker from inferring the new system parameters. However, majority of the above literature studies the effectiveness of MTD under DC and noiseless assumptions. As the detection rate of MTD is limited by the ratio between the attack strength and the noise level \cite{li2019feasibility}, there is no guarantee on the detection performance of existing MTD strategies against the unseen attacks in a noisy environment.

\subsection{Contributions}

With the attackers becoming more resourceful and intelligent, it is critical for the system operator to determine and guarantee the lowest detection rate of MTD against all unknown attacks. In this context, this paper introduces the concept of \textit{robust MTD}, which aims to guarantee the worst-case MTD effectiveness against a given level of attack strength under noisy environment. The main contributions of this paper are summarised as follows.

\begin{itemize}
    \item 
     This paper, for the first time, proposes the concept of robust MTD in a noisy environment. We theoretically prove that, for any given grid topology and MTD strategy, the minimal principal angle between the pre- and post-MTD Jacobian subspaces is directly linked with the worst-case performance against all potential attacks, which can be used as a new metric to represent the MTD effectiveness. 
    \item A novel MTD design algorithm is formulated to improve the worst-case detection rate by maximising the minimal principal angle under the complete grid configuration. We then demonstrate that the worst-case detection rate of the grid with incomplete configuration cannot be improved. Therefore, an iterative algorithm is formulated to maximise the minimal nonzero principal angle while limits the chance of attacking on the subspace that cannot be detected.
    \item Numerical simulations on IEEE case-6, 14, and 57 systems demonstrate the improved detection performance of robust MTD algorithms against the worst-case, random, and single-state attacks, under both simplified and full AC models. 
\end{itemize}

The rest of the paper is organised as follows. The preliminaries are summarised in Section \ref{sec:preliminary}; Analysis on MTD effectiveness is presented in Section \ref{sec:analysis_mtd_eff}; Problem formulation and proposed robust algorithms are presented in Section \ref{sec:robust mtd algorithm};
Case studies are given in Section \ref{sec:simulation} with conclusions in Section \ref{sec:conclusion}.

\section{Preliminaries}\label{sec:preliminary}

\subsection{Notations}

In this paper, vectors and matrices are represented by bold lowercase and uppercase letters, respectively. The $p$-norm of $\bm{a}$ is written as $\|\bm{a}\|_p$. The column space of $\bm{A}$ is $\mathcal{A}=\text{Col}(\bm{A})$. The kernel of a matrix $\bm{A}$ is represented as $\text{Ker}(\bm{A})$. The rank operator is written as \text{rank}($\bm{A}$). $\bm{P}_{\bm{A}} = \bm{A}(\bm{A}^T\bm{A})^{-1}\bm{A}^T$ represents the orthogonal projector to $\text{Col}(\bm{A})$ while $\bm{S}_{\bm{A}} = \bm{I} - \bm{P}_{\bm{A}}$ represents the orthogonal projector to $\text{Ker}(\bm{A}^T)$. The set of singular values is $\sigma(\bm{A}) = \{\sigma_1(\bm{A}), \sigma_2(\bm{A}),\dots,\sigma_{\min\{m,n\}}(\bm{A})\}$. The spectral norm is $\|\bm{A}\|_2 = \max_i \sigma_i(\bm{A})$ and the Frobenius norm is $\|\bm{A}\|_F$. We use the symbol $(\cdot)'$ to indicate the quantities after MTD and $(\cdot)_a$ to indicate the quantities after the attack. The matrix operator $\circ$ represents the Hadamard product. Other symbols and operators are defined in the paper whenever appropriate.

\subsection{System Model and State Estimation}

\label{sec:system_model}

\begin{figure}[]
    \centering
    \includegraphics[width=0.42\textwidth]{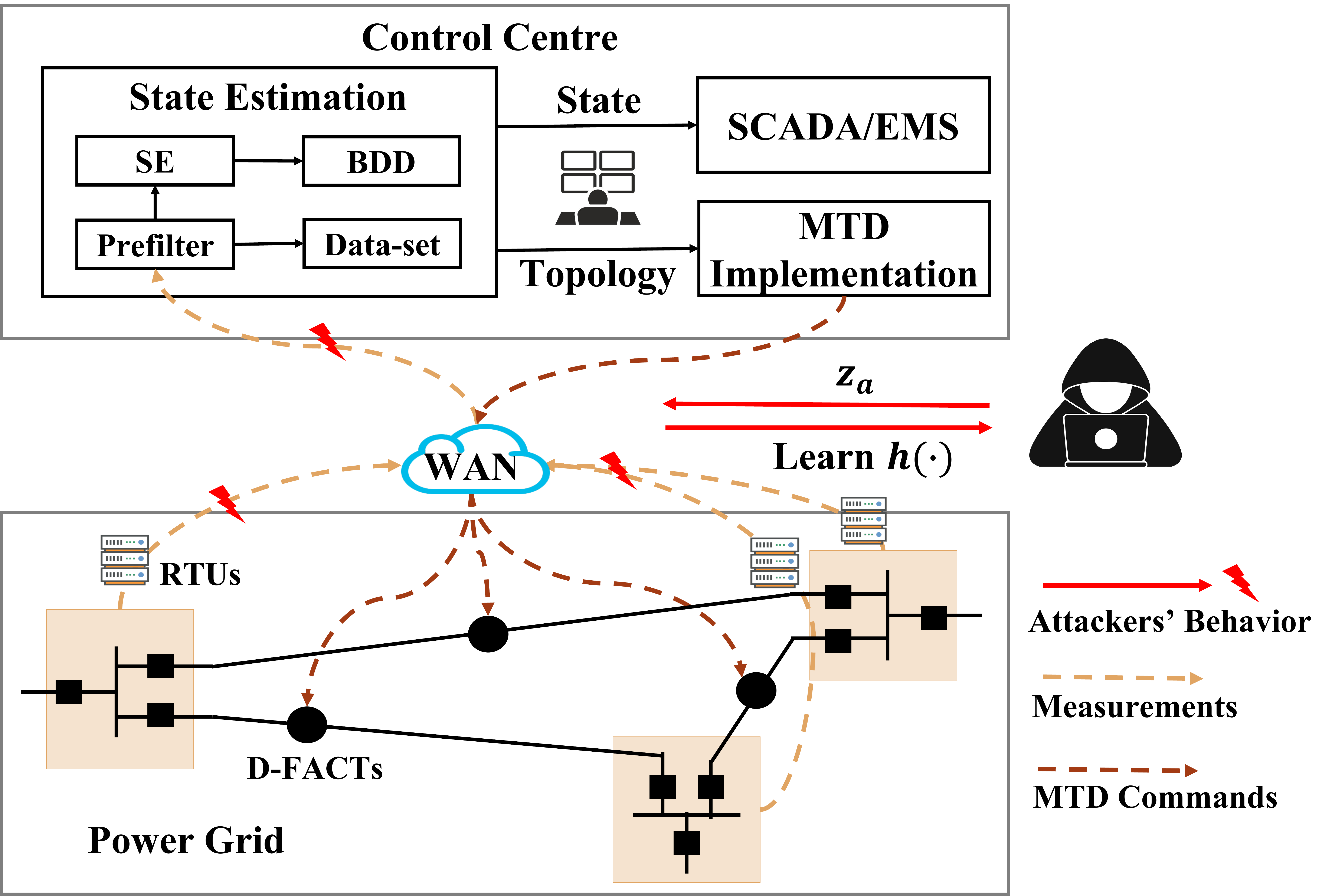}
    \caption{EMS with injection attacks and MTD in CPPS.}
    \label{fig:ems}
\end{figure}

The power system can be modelled as a graph $\mathcal{G}(\mathcal{N},\mathcal{E})$ with $|\mathcal{N}| = n+1$ number of buses and $|\mathcal{E}| = m$ number of branches. For each bus, we denote its complex voltage as ${\bm{\nu}}=\bm{v}\angle\bm{\theta}$; and for each branch, we denote the admittance as $\bm{y}=\bm{g}+j\bm{b}$. The power balances can be modelled by a set of nonlinear equations $\bm{z}=\bm{h}(\bm{\nu})+\bm{e}$ where $\bm{z}\in\mathbb{R}^{p}$ is the sensor measurement; $\bm{h}(\cdot)\in\mathbb{R}^{p}$ is the power balancing equation; $\bm{\nu}\in\mathbb{R}^{2n+1}$ is the system state composing of voltage magnitudes at all bus and phase angles at non-reference buses.  The measurement noise vector $\bm{e}\sim\mathcal{N}(\bm{0},\bm{R})$ follows an independent Gaussian distribution with diagonal covariance matrix $\bm{R}=\text{diag}([\sigma_1^2,\sigma_2^2,\cdots,\sigma_p^2])$. 

In detail, $\bm{h}(\cdot)$ is considered as \cite{gomez2018electric}:
\begin{subequations}\label{eq:power_balance}
\begin{equation}\label{eq:power_balance_pi}
    P_{i}=v_{i} \sum_{j=1}^{n} v_{j}\left(g_{i j} \cos \theta_{i j}+b_{i j} \sin \theta_{i j}\right)
\end{equation}
\begin{equation}\label{eq:power_balance_qi}
    Q_{i}=v_{i} \sum_{j=1}^{n} v_{j}\left(g_{i j} \sin \theta_{i j}-b_{i j} \cos \theta_{i j}\right)
\end{equation}
\begin{equation}\label{eq:power_balance_pf}
    P_{k: i \to j}=v_{i} v_{j}\left(g_{i j} \cos \theta_{i j}+b_{i j} \sin \theta_{i j}\right)-g_{i j} v_{i}^{2}
\end{equation}
\begin{equation}\label{eq:power_balance_qf}
    Q_{k: i \to j}=v_{i} v_{j}\left(g_{i j} \sin \theta_{i j}-b_{i j} \cos \theta_{i j}\right)+ b_{i j}v_{i}^{2}
\end{equation}
\end{subequations}
where $P_i$ and $Q_i$ are the active and reactive power injections at bus $i$; $P_{k:i\to j}$ and $Q_{k: i\to j}$ are the $k$-th active and reactive power flows from bus $i$ to $j$; $\theta_{ij} = \theta_i - \theta_j$ is the phase angle difference between bus $i$ and $j$.

As shown in Fig. \ref{fig:ems}, 
the control centre is equipped with state estimation (SE) which serves as a bridge between remote terminal units (RTU) and the energy management system (EMS) \cite{gomez2018electric}. Given the measurements, the AC-SE is solved by the following weighted least-square problem using iterative algorithm, such as Gauss-Newton method \cite{abur2004power}:
\begin{equation}\label{eq:ac_se}
    \min _{\hat{\bm{\nu}}} J(\hat{\bm{\nu}})=(\bm{z}-\bm{h}(\hat{\bm{\nu}}))^{T} \cdot \bm{R}^{-1} \cdot(\bm{z}-\bm{h}(\hat{\bm{\nu}}))
\end{equation}
where $\hat{\bm{\nu}}$ is the estimated state. Furthermore, the bad data detection (BDD) at the control centre detects any measurement error that violates a Gaussian prior. Given $\bm{\hat{\nu}}$, the residual vector is calculated as $\bm{r} = \bm{z}-\bm{h}(\bm{\hat{\nu}})$ and the residual is represented as $\gamma(\bm{z}) = \|\bm{R}^{-\frac{1}{2}}\bm{r}\|_2^2$. Let $\bm{e}$ be the random variable; then $\gamma$ approximately follows $\chi^2$ distribution with degree of freedom (DoF) $p-(2n+1)$ \cite{abur2004power}. The threshold $\tau_\chi(\alpha)$ of the $\chi^2$ detector can be defined probabilistically based on the desired False Positive Rate (FPR) $\alpha\in(0,1)$ by the system operator \cite{abur2004power}:
\begin{equation}\label{eq:chi_square_distribution}
    \int_{\tau_\chi(\alpha)}^\infty g(u)du = \alpha
\end{equation}
where $g(u)$ is the p.d.f of the $\chi^2$ distribution and $\alpha$ is usually set as 1\%-5\%. Consequently, the BDD detector can be designed as:
\begin{equation*}\label{eq:detector_bdd}
\mathcal{D}_{BDD}(\bm{z}) =  \begin{cases} 1 & \gamma(\bm{z}) \geq \tau_\chi(\alpha) \\ 0 & \gamma(\bm{z}) < \tau_\chi(\alpha)\end{cases}
\end{equation*}

\subsection{Attack Assumptions}
\label{sec:fdia}

With the emerging implementation of information and communication techniques, standard protocols, such as Modbus, can be vulnerable to FDI attacks. It has been shown that an FDI attack $\bm{z}_a=\bm{z}+\bm{a}$ can bypass the BDD if $\bm{a} = \bm{h}({\bm{\nu}}+\bm{c}) - \bm{h}({\bm{\nu}})$ where $\bm{c}$ is the attack vector on the state vector.
In this case, the contaminated measurement becomes $\bm{z}_a=\bm{h}(\bm{\nu}+\bm{c})+\bm{e}$ whose residual follows the same $\chi^2$ distribution as the legit measurement $\bm{z}$. 

To successfully launch FDI attacks, we assume the attacker's abilities as follows.

\textbf{\textit{Assumption 1}}: The attackers can access all measurements and are aware of the admittance and topology of the grid to build $\bm{h}(\cdot)$. The exfiltration can be achieved by data-driven algorithms \cite{zhang2020topology, kim2014subspace, yu2015blind, lakshminarayana2021data}. However, the duration of data collection is much longer than a single state estimation time, implying that the attacker cannot immediately know the exact value of reactance changes \cite{lakshminarayana2021cost}. Meanwhile, attackers are also aware of the exact state or estimation of the state from previous measurements \cite{hug2012vulnerability, rahman2013false}.

\textbf{\textit{Assumption 2}}: The attackers can modify or replace all the eavesdropped measurements to achieve their purposes. However, since large instant measurement changes may violate the temporal trends of the grid measurements and be detected \cite{zhao2017short, xu2020deep}, the attack strength $\|\bm{a}\|_2$ is assumed to be small. 

Assumptions 1-2 require the attacker's efforts to gain sufficient knowledge on the grid topology and operational conditions, which may not be easy in practise. However, we assume a strong attack ability and study the defence algorithm against general and unpredictable FDI attacks.

\subsection{Moving Target Defence}

By using the D-FACTS devices, the system operator can proactively change the reactances to keep invalidating the attacker's knowledge on $\bm{h}(\cdot)$:
\begin{equation*}
    \bm{h}_{\bm{x}}(\cdot)\xrightarrow{\text{D-FACTS}} \bm{h}_{\bm{x}'}(\cdot)
\end{equation*}
where $\bm{x}' = \bm{x} + \Delta \bm{x}$ is the reaction after activating the D-FACTS devices. As illustrated in Fig. \ref{fig:ems}, the channels of D-FACTS devices are encrypted and MTD is implemented with a period shorter than the reconnaissance time of the attacker (see Assumption 1). In addition, the reactances changed by the D-FACTS devices are physically limited:
\begin{subequations}
\begin{equation}\label{eq:reactance_limit_in}
    -\tau \bm{x}_i \leq  \Delta \bm{x}_i \leq \tau \bm{x}_i,\quad i\in\mathcal{E}_D
\end{equation}
\begin{equation}\label{eq:reactance_limit_out}
    \Delta \bm{x}_i = 0, \quad i\in\mathcal{E}\setminus\mathcal{E}_D
\end{equation}
\end{subequations}
where $\bm{x}_i$ is the reactance of the $i$th branch; $\tau$ represents the maximum perturbation ratio of D-FACTS devices. Typical values of $\tau$ are reported as $20\%-50\%$ in the literature \cite{liu2018reactance,liu2020optimal,zhang2019analysis,lakshminarayana2021cost}; $\mathcal{E}_D$ represents the set of branches equipped with the D-FACTS devices.
After implementing MTD, the residual vector becomes $\bm{r}_a' = \bm{h}'(\bm{x}) + \bm{h}(\bm{x}+\bm{c}) - \bm{h}(x) + \bm{e}$ which may no longer follow the $\chi^2$ distribution of the legit measurement and hence trigger the BDD. 

\subsection{Model Simplification for MTD Design}

To design the MTD against FDI attacks, most of the literature relies on DC or simplified AC power system models \cite{liu2018reactance, zhang2019analysis, liu2020optimal, lakshminarayana2021cost, zhang2021strategic, li2019feasibility, liu2019joint} and then verifies the performance on the full AC model. Here, we adopt the simplified AC model based on the linearised measurement equation. Compared with the DC model, the simplified AC model can reflect different state values with branch resistance also considered. 

In detail, the first-order Taylor expansion can be established around a stationary state $\bm{\nu}_0$:
\begin{equation}\label{eq:linear_equation}
    \bm{z} = \bm{h}(\bm{\nu}_0) + \bm{J}_{\bm{\nu}_0}(\bm{\nu}-\bm{\nu}_0) + \bm{e}
\end{equation}
where the Jacobian matrix of $\bm{h}(\cdot)$ is found with respect to $\bm{\nu}_0$ as $\bm{J}_{\bm{\nu}_0} = \left[\left.\frac{\partial \bm{h}_k}{\partial \bm{\nu}_i} \right\vert_{\bm{\nu}=\bm{\nu}_0} \right]_{i,k}$. The state $\bm{\nu}_0$ can be simulated from security constrained AC-OPF \cite{gomez2018electric} around the estimated active and reactive loads before the real-time operation. Alternatively, the states estimated from the previous measurements or a flat state \cite{liu2019joint, liu2021optimalcoding} can also be used. Following the recent literature on MTD \cite{liu2018reactance,liu2020optimal,liu2019joint}, 
we consider the FDI attacks on the voltage phase angle and derive the defence strategies according to the power flow measurements at each branch. Therefore, the Jacobian matrix is considered as follows.
\begin{equation}\label{eq:jacobian}
    \bm{J}_{\bm{\theta_0}} = \left[\left.\frac{\partial \bm{P}_{k:i\to j}}{\partial \bm{\theta}_i} \right\vert_{\bm{\theta}=\bm{\theta}_0} \right]_{k} = -\bm{V}\cdot\bm{G}\cdot\bm{A}_r^{\sin}+\bm{V}\cdot\bm{B}\cdot\bm{A}_r^{\cos}
\end{equation}
where $\bm{V} = \text{diag}\left((\bm{C}_f\bm{v})\circ(\bm{C}_t\bm{v})\right)$; $\bm{G} = \text{diag}(\bm{g})$; $\bm{B} = \text{diag}(\bm{b})$; $\bm{A}_{r}^{\sin} = \text{diag}(\sin{\bm{A\theta_0}})\bm{A}_r$; and $\bm{A}_{r}^{\cos} = \text{diag}(\cos{\bm{A\theta_0}})\bm{A}_r$. Moreover, $\bm{C}_f$ and $\bm{C}_t$ are the `from' and `to' -side incidence matrices; $\bm{A}_r$ is the reduced incidence matrix by removing the column representing the reference bus from the incidence matrix $\bm{A}$. To simplify the notation, we omit the subscript $\bm{\theta}_0$ in $\bm{J}_{\bm{\theta_0}}$ in the following discussion.

According to Assumption 2, as the attack strength is limited, the attack vector can also be linearised around $\bm{\theta}_0$ as \cite{liu2019joint}:
\begin{equation}\label{eq:linear_attack}
    \bm{a} = \bm{h}(\bm{\theta}_0+\bm{c}) - \bm{h}(\bm{\theta}_0) = \bm{J}\bm{c}
\end{equation}

We design the MTD algorithm based on the simplified AC model \eqref{eq:linear_equation}-\eqref{eq:linear_attack} using active power flow measurements. The proposed MTD will be applied to the original AC model \eqref{eq:power_balance}-\eqref{eq:ac_se} in the simulation.

\section{Analysis on MTD Effectiveness}

\label{sec:analysis_mtd_eff}

In this section, we first extend the concept of complete MTD in the literature from DC model to simplified AC model. We then define the MTD effectiveness in a probabilistic way and illustrate the need for a new metric on effective MTD design in a noisy environment.

\subsection{Complete MTD}

Let $\bm{H}$ and $\bm{H}'$ be the DC measurement matrices. Under the noiseless condition, the \textit{complete MTD} can be designed to detect any FDI attack by keeping the composite matrix $[\bm{H},\bm{H}']$ full column rank \cite{liu2018reactance, zhang2019analysis, liu2020optimal, zhang2021strategic}. If the full rank condition cannot be achieved due to the sparse grid topology (e.g. $m<2n$) or limited number of D-FACTs devices, a \textit{max-rank incomplete MTD} can be designed to minimise the attack space. As the rank of the composite matrix is maximised under both complete and incomplete conditions, we refer to the MTD strategies in \cite{liu2018reactance, zhang2019analysis, liu2020optimal, zhang2021strategic} as \textit{max-rank MTD}.  

To better define the problem, we extend the concept of complete and incomplete MTDs from the DC model to the simplified AC models in the following proposition:

\begin{proposition}\label{prop:com_incom}
    The power system modelled by \eqref{eq:linear_equation} is with complete configuration against the FDI attack modelled by \eqref{eq:linear_attack} only if $m\geq 2n$ where $m$ and $n$ are the number of branches and the number of non-reference buses, respectively.
\end{proposition}

\begin{proof}
Please refer to Appendix \ref{sec:proof_com_incom}.
\end{proof}

As stated by Proposition \ref{prop:com_incom}, to have a complete configuration $\text{rank}([\bm{J}_N,\bm{J}_N']) = 2n$, the number of branches should be at least one time larger than the number of non-reference buses. In addition, the max-rank incomplete MTD with $\text{rank}([\bm{J}_N,\bm{J}_N']) = m$ can be designed for the grid with incomplete configuration. In the following discussions, we refer the grid that can achieve complete MTD under certain topology and D-FACTS device deployment as \textit{complete configuration}, otherwise as \textit{incomplete configuration}. 

\subsection{$\beta$-Effective MTD}

Following \eqref{eq:linear_equation}, denote $\bm{z} \triangleq \bm{z} - \bm{h}(\bm{\theta}_0)$ and $\bm{\theta} \triangleq \bm{\theta} - \bm{\theta}_0$. For the new system equation $\bm{z} = \bm{J}\bm{\theta} + \bm{e}$, the residual vector of the $\chi^2$ detector can be written as $\bm{r} = \bm{S}(\bm{J}\bm{\theta}+\bm{e}) = \bm{Se}$ where $\bm{S} = \bm{I} - \bm{J}(\bm{J}^T\bm{R}^{-1}\bm{J})^{-1}\bm{J}^T\bm{R}^{-1}$ is the weighted orthogonal projector on $\text{Ker}(\bm{J}^T)$. The residual $\gamma = \|\bm{R}^{-\frac{1}{2}}\bm{S}\bm{e}\|_2^2$ follows the $\chi^2$ distribution with DoF $m-n$. Referring to the simplified attack model \eqref{eq:linear_attack}, the residual vector after MTD under attack can be written as $\bm{r}_a' = \bm{S}'(\bm{J}\bm{c} + \bm{e})$ where $\bm{S}' = \bm{I} - \bm{J}'(\bm{J}'^T\bm{R}^{-1}\bm{J}')^{-1}\bm{J}'^T\bm{R}^{-1}$. As $\bm{a}$ is usually not in $\mathcal{J}'$ and $\bm{r}'_a$ is biased from zero, 
the residual $\gamma_a' = \|\bm{R}^{-\frac{1}{2}}\bm{S}'(\bm{Jc} + \bm{e})\|_2^2$ follows the non-central $\chi^2$ distribution, i.e. $\gamma_a' \sim \chi_{m-n}^2(\lambda)$ with non-centrality parameter $\lambda = \|\bm{R}^{-\frac{1}{2}}{\bm{S}}'\bm{Jc}\|_2^2$ \cite{krishnamoorthy2006handbook}. Meanwhile, the mean and variance of the distribution are given as $\mathbf{E}(\gamma_a') = m-n +\lambda$ and $\mathbf{Var}(\gamma_a') = 2(m-n+2\lambda)$, respectively. For clear presentation, the matrices are normalised with respect to the measurement noises, e.g., $\bm{J}_N = \bm{R}^{-\frac{1}{2}}\bm{J}$ and $\bm{a}_N = \bm{J}_N\bm{c}$. More details can be found in Appendix \ref{sec:app_normalize}.  

It is clear that when a noisy environment is considered, deterministic criteria can no longer be used to describe the effectiveness of MTD. A probabilistic criteria is hence defined. Following \eqref{eq:chi_square_distribution}, for any given attack vector $\bm{a}$, we define an MTD as $\beta$-\textit{effective} ($\beta$-MTD in short) if the following inequality is satisfied:
\begin{equation}\label{eq:non_central_chi_square_distribution}
    f(\lambda) = \int_{\tau_\chi(\alpha)}^\infty g_\lambda(u)du \geq \beta
\end{equation}
where $g_\lambda(u)$ is the p.d.f. of non-central $\chi^2$ distribution and $\beta\in(0,1)$ is a desired detection rate. When $\lambda$ increases from 0, the detection probability on $\bm{a}$ also increases as the mean and variance increase \cite{teixeira2010cyber}. Therefore, for a given $\beta$, there exists a minimum $\lambda$ such that \eqref{eq:non_central_chi_square_distribution} is satisfied. This minimum $\lambda$ is defined as critical and denoted as $\lambda_c(\beta)$.

Consequently, the rank conditions in \cite{liu2018reactance, zhang2019analysis, liu2020optimal, zhang2021strategic, liu2019joint} cannot guarantee detection performance, as they are not directly linked with the increase of $\lambda$ to have $\beta$-MTD. Fig. \ref{fig:residual_illustration} illustrates the c.d.f. of $\gamma'$ on a random FDI attack using max-rank MTDs in a case-14 system. Without using MTD, the detection rate is 5\% which is consistent with the FPR. To have a high detection rate, e.g., $\beta=95\%$, it is desirable to sufficiently shift the distribution as shown by the blue curve. The max-rank MTDs can shift the c.d.f. positively, but there is no guarantee on how much of such shift can be achieved and whether it leads to the desired detection rates. This finding clearly calls for a new design of MTD algorithm in a noisy environment. 

Moreover, as numerically shown by \cite{li2019feasibility}, not all attacks can be detected by the MTD with high detection rate. Therefore, we theoretically introduce the following necessary condition to have $\beta$-MTD which can be seen as the limitation of MTD against FDI attacks with small attack strength.

\begin{proposition}\label{prop:min_att}
An MTD is $\beta$-effective only if $\|\bm{a}_N\|_2\geq \sqrt{\lambda_c(\beta)}$. 
\end{proposition}

\begin{proof}
Please refer to Appendix \ref{sec:app_min_att}.
\end{proof}

Proposition \ref{prop:min_att} can be further analysed on $\bm{a}$ to have $\|\bm{a}\|_2\geq \sigma_{min}\sqrt{\lambda_c(\beta)}$ with $\sigma_{min} = \min_i\{\sigma_1,\sigma_2,\dots,\sigma_m\}$. This implies that $\beta$-MTD can be achieved only if the ratio between attack strength and measurement noise is higher than a certain value, which verifies the numerical results in \cite{li2019feasibility}. 

\begin{figure}[]
    \centering
    \includegraphics[width=0.24\textwidth]{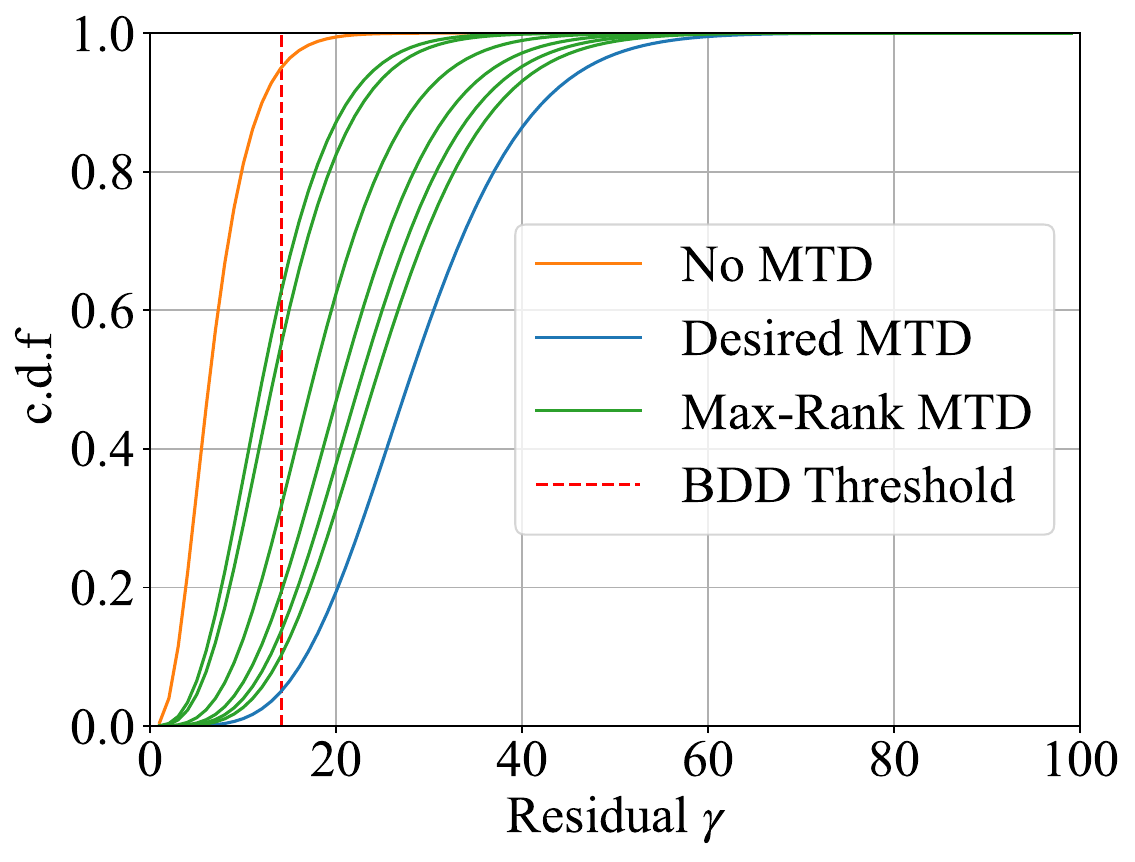}
    \caption{Illustration of attack detection probability on IEEE case-14 system based on simplified AC model \eqref{eq:linear_equation}-\eqref{eq:linear_attack}. The more positively the c.d.f. is shifted, the higher averaged detection rate can be achieved.}
    \label{fig:residual_illustration}
\end{figure}

\subsection{Max MTD}
While Proposition \ref{prop:min_att} establishes the theoretical limit on the detection probability for any given attack strength, in practise, the constraints on D-FACTS devices \eqref{eq:reactance_limit_in}-\eqref{eq:reactance_limit_out} further restricts such limit. In this context, the maximum detection rate on a known attack vector $\bm{a}_N$, with the limits of the D-FACTS devices considered, can be found by the \textit{max-MTD} algorithm:

\begin{equation}\label{eq:max_mtd}
\begin{array}{cc}
      \max_{\Delta \bm{x}} & {\|\bm{S}_N'\bm{a}_N\|_2^2}  \\
    \text{s.t.} & \eqref{eq:reactance_limit_in}-\eqref{eq:reactance_limit_out}
\end{array}
\end{equation}

In practice, it is impossible to design $\Delta \bm{x}$ to achieve a certain $\lambda_c(\beta)$ in advance as $\bm{a}_N$ cannot be known. Nonetheless, max-MTD can be regarded as the performance upper-bound for any MTD strategy with the same placement and perturbation limit.

\section{Robust MTD Algorithms}

In this section, we start by establishing the concept of robust MTD and its mathematical formulation. Then the robust MTD algorithms are formulated for the grid with complete and incomplete configurations, respectively.

\label{sec:robust mtd algorithm}

\subsection{Definition and Problem Formulation}

Instead of considering the average detection rate, this paper defines the robust MTD that can maximise the worst-case detection rate against all possible attacks. First, we define the weakest point for a given MTD design as follows.

\begin{definition}\label{def:vul}
Given $\Delta \bm{x}$ and the corresponding pair of subspaces $(\mathcal{J}_N,\mathcal{J}_N')$, the weakest point of $(\mathcal{J}_N,\mathcal{J}_N')$ is defined as a unitary element $\bm{j}_N^*\in\mathcal{J}_N$ such that $\lambda(\Delta \bm{x}, {\bm{j}_N^*}) \leq \lambda(\Delta \bm{x},\bm{j}_N)$ for $\forall \bm{j}_N\in\mathcal{J}_N$, $\|\bm{j}_N\|_2=1$. The worst-case detection rate for attack strength $\|\bm{a}_N\|_2 = |a|\neq0$ is defined as $f(\lambda_{\text{min}})$ with $\lambda_{\text{min}} = \lambda(\Delta \bm{x}, a\bm{j}_N^*)$.
\end{definition}

According to the Definition \ref{def:vul}, the weakest point in $(\mathcal{J}_N,\mathcal{J}_N')$ satisfies $|a|\|\bm{S}_N'\bm{j}_N^*\|_2\leq |a|\|\bm{S}_N'\bm{j}_N\|_2$, $\forall \bm{j}_N\in\mathcal{J}_N, \|\bm{j}_N\|_2 = 1, a \neq 0$. Let $\bm{a}_N^* = a\bm{j}_N^*$ and $\bm{a}_N = a\bm{j}_N$, the detection rate on $\bm{a}_N^*$ is the lowest among all attacks with the same strength as $\|\bm{S}_N'\bm{a}_N^*\|_2\leq \|\bm{S}_N'\bm{a}_N\|_2, \forall \bm{a}_N\in\mathcal{J}_N, \|\bm{a}_N\|_2=|a| \neq 0$.
Note that the weakest point may not be unique, but all of them have the same worst-case detection rate.

Based on the definition of MTD weakest point, the following robust max-min optimization problem can be formulated:
\begin{equation}\label{eq:robust_opt}
\begin{array}{cc}
      \max_{\Delta \bm{x}} \min_{\|\bm{a}_N\|_2 = 1,\bm{a}_N \in \mathcal{J}_N} & {\|\bm{S}_N'\bm{a}_N\|_2^2}  \\
    \text{s.t.} & \eqref{eq:reactance_limit_in}-\eqref{eq:reactance_limit_out}
\end{array}
\end{equation}

The inner problem $\min_{\|\bm{a}_N\|_2 = 1,\bm{a}_N \in \mathcal{J}_N} {\|\bm{S}_N'\bm{a}_N\|_2^2}$ is the mathematical formulation of the weakest point in Definition \ref{def:vul} which is maximised over the outer programming. From a game-theoretic point of view, we can present this setting as an intelligent attacker aims to develop an FDI attack with the highest probability to bypass BDD and the system operator tries to improve his/her defence strategy against this intelligent attacker. 

In the following sections, we will show that the two-stage problem \eqref{eq:robust_opt} can be reduced into a single-stage minimisation problem by analytically representing the weakest point using the principal angles between $\mathcal{J}_N$ and $\mathcal{J}_N'$.

\subsection{Robust MTD for the Grid with Complete Configuration }

Similar to the one-dimensional case where the angle between two unitary vectors $\bm{u}$ and $\bm{v}$ is defined as $\cos{\theta} = \bm{v}^T\bm{u}$, the minimal angle between subspaces $\mathcal{J}_N, \mathcal{J}_N'\subseteq \mathbb{R}^p$ is defined as $0 \leq \theta_{1} \leq \pi/2$ \cite{meyer2000matrix}:
\begin{equation}\label{eq:minimal_angle}
    \cos \theta_{1}=\max_{\bm{u}\in \mathcal{J}_N, \bm{v} \in \mathcal{J}_N' \atop\| \bm{u}\left\|_{2}=\right\| \bm{v} \|_{2}=1} \bm{u}^{T} \bm{v} = \bm{u}_1^{T} \bm{v}_1
\end{equation}
where $\theta_1$ is the minimal principal angle; $\bm{u}_1$ and $\bm{v}_1$ are the first principal vectors. Referring to \eqref{eq:minimal_angle}, the following proposition specifies that the weakest point with the lowest detection rate of $(\mathcal{J}_N,\mathcal{J}_N')$ 
is the first principal vector $\bm{u}_1$ associated with the minimal principal angle $\theta_1$.

\begin{proposition}\label{prop:min_eff}
Given a pair of $(\mathcal{J}_N,\mathcal{J}_N')$, the minimum non-centrality parameter under attack strength $\|\bm{a}_N\|_2 = |a| \neq 0$ is $\lambda_{\text{min}} = a^2\sin^2{\theta_1}$. Meanwhile, $\lambda_{\text{min}}$ is achieved by attacking the first principal vector $\bm{u}_1$ of $\bm{J}_N$.
\end{proposition}

\begin{proof}
Please refer to Appendix \ref{sec:app_proof_1}. 
\end{proof}

When $\theta_1 = \pi/2$, Proposition \ref{prop:min_eff} implies that the minimum non-centrality parameter is equal to $a^2$. As two subspaces are orthogonal if $\theta_1 = \pi/2$, Proposition \ref{prop:min_eff} is consistent with the maximum detection probability stated in Theorem 1 of \cite{lakshminarayana2021cost}.

In addition, as $\sin{\cdot}$ is monotonically increasing in $[0,\pi/2]$, Proposition \ref{prop:min_eff} demonstrates that the two-stage problem \eqref{eq:robust_opt} can be equivalently solved by one-stage maximisation:
\begin{equation}\label{eq:robust_opt_2}
    \begin{array}{cc}
         \max_{\Delta\bm{x}} & \theta_1 \\
        \text{s.t.} & \eqref{eq:reactance_limit_in}-\eqref{eq:reactance_limit_out}
    \end{array}
\end{equation}

To analytically represent $\theta_1$, a sequence of principal angles $\Theta = \{\theta_1,\theta_2,\dots,\theta_n\}$ can be defined iteratively by finding the orthonormal basis of $\mathcal{J}_N$ and $\mathcal{J}_N'$ such that for $i=2,\dots,n$ \cite{meyer2000matrix}:
\begin{equation}\label{eq:principal_angle}
            \cos \theta_{i}=\max _{\bm{u} \in \mathcal{J}_{N,i}, \bm{v} \in \mathcal{J}_{N,i}' \atop\|\bm{u}\|_{2}=\|\bm{v}\|_{2}=1} \bm{u}^{T} \bm{v}=\bm{u}_{i}^{T} \bm{v}_{i}
\end{equation}
where $\mathcal{J}_{N,i} = \bm{u}_{i-1}^{\perp} \cap \mathcal{J}_{N,i-1}$ and $\mathcal{J}_{N,i}'=\bm{v}_{i-1}^{\perp} \cap \mathcal{J}_{N,i-1}'$.

$\Theta$ can be separated into three parts. 
Let $\Theta_{1} = \{\theta_i|\theta_i = 0\}$,  $\Theta_{2} = \{\theta_i|0<\theta_i < \pi/2\}$, and $\Theta_{3} = \{\theta_i|\theta_i = \pi/2\}$ with cardinality equal to $k$, $r$, and $l$, respectively, and $n=k+r+l$.
The corresponding vectors $\bm{U} =  \{\bm{u}_1,\bm{u}_2,\dots,\bm{u}_n\}$ and $\bm{V} = \{\bm{v}_1,\bm{v}_2,\dots,\bm{v}_n\}$ are called principal vectors, which are the orthonormal basis of $\mathcal{J}_N$ and $\mathcal{J}_N'$, respectively. Similarly, $\bm{U}$ and $\bm{V}$ can also be separated into $\bm{U}_{1}, \bm{V}_{1},\cdots$. Specifically, $\bm{U}_1=\bm{V}_1=\mathcal{J}_N'\cap\mathcal{J}_N$ represents the intersection subspace of dimension $k$ and $l$ is the dimension of orthogonality. Furthermore, it is proved that there always exist semi-orthogonal matrices $\bm{U}$ and $\bm{V}$ for any $\mathcal{J}_N$ and $\mathcal{J}_N'$ such that the bi-orthogonality is satisfied \cite{galantai2008subspaces}:
\begin{equation}\label{eq:bithogonality}
    \bm{U}^T\bm{V} = \text{diag}([\cos{\theta_1}, \cos{\theta_2}, \dots, \cos{\theta_n}]) = \Gamma
\end{equation}

Since the orthogonal projector is uniquely defined \cite{meyer2000matrix} and also by \eqref{eq:bithogonality}, rewriting $\bm{P}_N = \bm{U}\bm{U}^T$ and $\bm{P}_N' = \bm{V}\bm{V}^T$ gives
\begin{equation}\label{eq:com_svd}
        \bm{P}_N\bm{P}_N' =\bm{U}\bm{U}^T\bm{V}\bm{V}^T =\bm{U}\Gamma \bm{V}^T
\end{equation}

Eq. \eqref{eq:com_svd} is the truncated singular value decomposition (t-SVD) on $\bm{P}_N\bm{P}_N'$ where the diagonal matrix $\Gamma$ contains the first $n$ largest singular values of $\bm{P}_N\bm{P}_N'$, and $\bm{U}$ and $\bm{V}$ are the first (left- and right-hand) $n$ singular vectors of $\bm{P}_N\bm{P}_N'$. As $\sigma(\bm{P}_N\bm{P}_N') = \{\bm{1}_k, \cos^2{\theta_{k+i}(i=1,\dots,r)},\bm{0}_{k+r+i}(i=1,\dots,l),\bm{0}_{n+i}(i=1,\dots,m-n)\}$, this t-SVD is an exact decomposition of $\bm{P}_N\bm{P}_N'$. 


Based on the t-SVD, Algorithm \ref{alg:vulnerable_point} is proposed to find the weakest point and the worst-case detection rate. For the grid with complete configuration, the composite matrix can be full column rank so that $k=0$. Line 6 outputs the weakest point $\bm{u}_1$ while line 9 outputs the empty intersection subspace. The worst-case detection rate is calculated according to Proposition \ref{prop:min_eff} in line 7. Practically, once the MTD strategy is determined, the weakest point $\bm{u}_1$ of this strategy can be directly spotted. Therefore, the system operator can evaluate the worst-case detection rate with respect to a maximum tolerable attack strength $|a|$. 

\begin{algorithm}[t]
    \footnotesize
    \SetKwInOut{Input}{Input}
    \SetKwInOut{Output}{Output}

    \Input{grid topology $\mathcal{G}(\mathcal{N},\mathcal{E})$, reactance perturbation $\Delta \bm{x}$, and attack strength $|a|$}
    \Output{weakest point $\bm{u}_{k+1}$, intersection subspace $\bm{U}_{1}$, and worst-case detection rate $f_{min}$}
    Construct the pre- and post- MTD measurement matrices $\bm{J}_N$ and $\bm{J}_N'$, respectively;
    
    Find the orthogonal projectors $\bm{P}_N$ and $\bm{P}_N'$ on $\bm{J}_N$ and $\bm{J}_N'$. Then do t-SVD \eqref{eq:com_svd};
    
    $rank = \text{rank}([\bm{J}_N,\bm{J}_N'])$; \tcc{Rank of the composite matrix.}
    
    $k = 2n - rank$; \tcc{The dimension of $\mathcal{J}_N'\cap\mathcal{J}_N$.}
    
    $\cos(\theta_{k+1}) = \Gamma(k+1,k+1)$;
    
    $\bm{u}_{k+1} = \bm{U}(k+1,k+1)$;  \tcc{The weakest point in $\mathcal{J}_N\setminus(\mathcal{J}_N'\cap \mathcal{J}_N)$.}
    
    $f_{\min} = f(a^2\sin^2(\theta_{k+1}))$; \tcc{The worst-case detection rate in $\mathcal{J}_N\setminus(\mathcal{J}_N'\cap \mathcal{J}_N)$.}
    
    \eIf{$rank = 2n$}
    {
    $\bm{U}_{1} = \varnothing$; \tcc{Complete MTD configuration.}
    }
    {
    $\bm{U}_{1} = \bm{U}(:,1:k)$; \tcc{Incomplete MTD configuration.}
    }
    \caption{Find the Weakest Point(s) and the Worst-Case Detection Rate}
    \label{alg:vulnerable_point}
\end{algorithm}

The t-SVD \eqref{eq:com_svd} also results in a solvable reformulation of \eqref{eq:robust_opt_2}. The worst-case detection rate can be maximised by the \textit{robust MTD} algorithm for the grid with complete configuration as follows:
\begin{equation}\label{eq:complete_mtd}
    \begin{array}{cc}
         \min_{\Delta\bm{x}} & {\|\bm{P}_N\bm{P}_N'\|_2} \\
        \text{s.t.} & \eqref{eq:reactance_limit_in}-\eqref{eq:reactance_limit_out}
    \end{array}
\end{equation}
where the property $\|\bm{P}_N\bm{P}_N'\|_2 = \sigma_{\max}(\bm{P}_N\bm{P}_N') = \cos(\theta_1)$ is used and $\|\bm{P}_N\bm{P}_N'\|_2\in[0,1]$. 

\begin{remark}
The robust MTD algorithm \eqref{eq:complete_mtd} requires sufficient placement of D-FACTS devices (as a planning stage problem) to guarantee $k=0$, e.g., using the `D-FACTS placement for the complete MTD' algorithm in \cite{liu2020optimal}.
\end{remark}

\subsection{Robust MTD for the Grid with Incomplete Configuration}

The robust MTD in \eqref{eq:complete_mtd} is not tractable for power system with incomplete MTD configuration. As $k\neq0$, $\theta_1 \equiv 0$ and $\|\bm{P}_N\bm{P}_N'\|_2 \equiv 1$ no matter how $\Delta \bm{x}$ is designed. Fig. \ref{fig:illustraion_3d} shows a three-dimensional incomplete-MTD case. The attack $\bm{a}_N$ in green shows a random attack attempt with non-zero $\lambda$. However, the weakest point $\text{Col}(\bm{u_1})$ is not trivial. As the attacker can possibly target $\text{Col}(\bm{u}_1)$, the worst-case detection rate is constantly equal to FPR. In addition to $\theta_1$, every attack in $\bm{U}_1$ is undetectable. The intersection can be regarded as the space of the weakest points, whose dimension is calculated as $k = 2n-\text{rank}([\bm{J}_N,\bm{J}_N'])\neq 0$. Therefore, the smallest non-zero principal angle (which also corresponds to the weakest point in $\mathcal{J}_N\setminus (\mathcal{J}_N'\cap\mathcal{J}_N)$) can be found as $\theta_{k+1}$ in line 5 of Algorithm \ref{alg:vulnerable_point} with the minimum detection rate calculated in line 7. Meanwhile, $\bm{U}_1$, corresponding to the subspace that cannot be detected, is calculated in line 11. 
\begin{figure}
    \centering
    \includegraphics[width=0.24\textwidth]{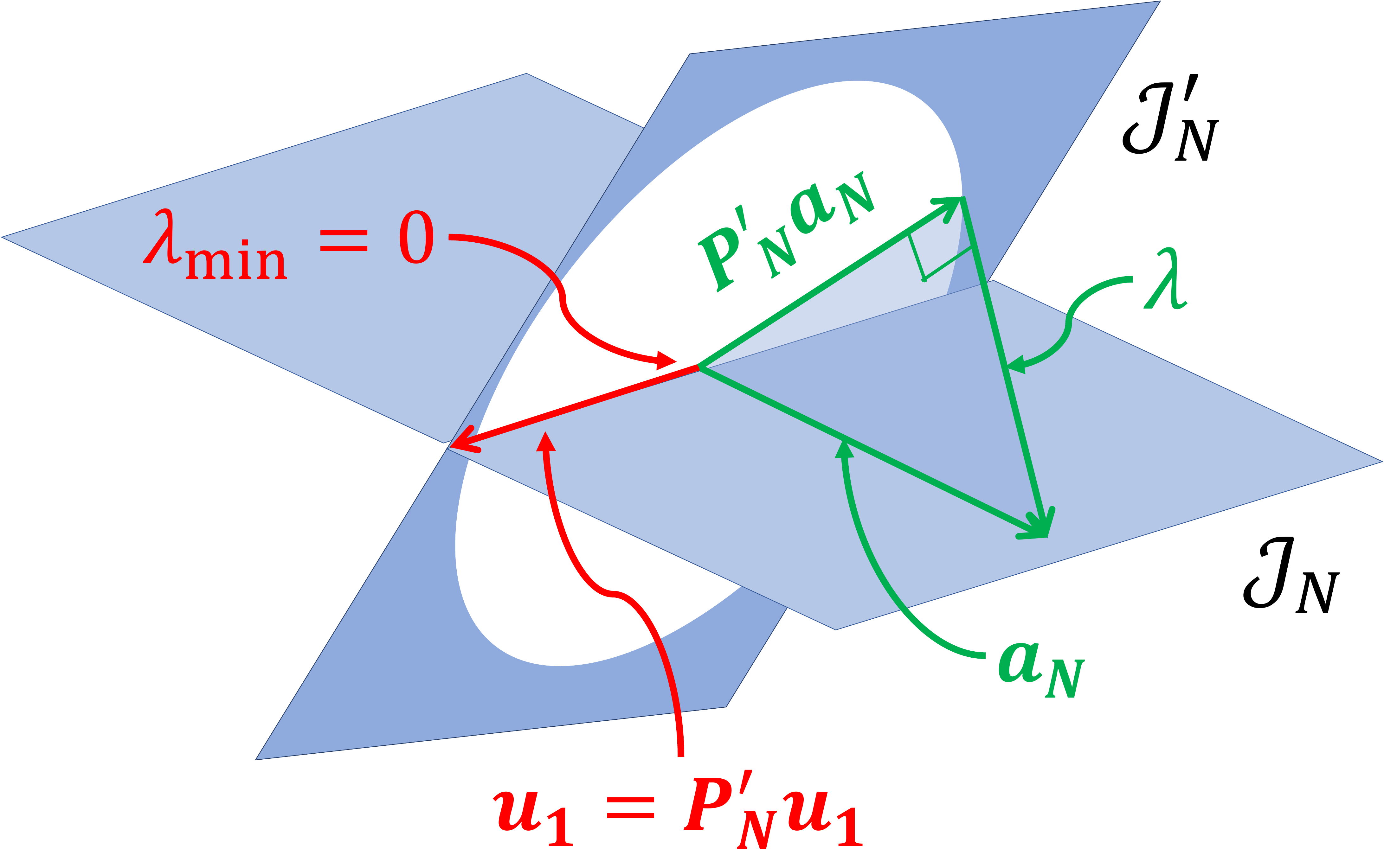}
    \caption{An illustration on the grid with incomplete configuration, $\mathcal{J}_N,\mathcal{J}_N'\subset \mathbb{R}^3$.}
    \label{fig:illustraion_3d}
\end{figure}

To solve the intractable problem, the following design principles are considered which can improve the robust performance of MTD with incomplete configuration: 

\textbf{\textit{Principle 1}}: Minimise $k$, the dimension of the intersection.

\textbf{\textit{Principle 2}}: The attacker shall not easily attack on the intersection subspace $\bm{U}_1$ by chance.

\textbf{\textit{Principle 3}}: Maximise $\theta_{k+1}$, the minimum nonzero principal angle in $(\mathcal{J}_N,\mathcal{J}_N')$. 

Each of the principles is discussed as follows.

\textbf{\textit{Principle 1:}} The idea of Principal 1 is to minimise the attack space that can never be detected by MTD so that the probability of detectable FDI attacks increases. Minimising $k$ is a planning stage problem as the rank of the composite matrix is almost not related to the perturbation amount of the D-FACTS devices once they have been deployed \cite{zhang2019analysis}. 
In this paper, we propose a new D-FACTS device placement algorithm to achieve the minimum $k$. Compared with the existing work \cite{liu2020optimal, zhang2019analysis, liu2018reactance}, our algorithm uses the BLOSSOM algorithm \cite{galil1986efficient} to find the maximum cardinality matching \cite{bondy2008graph} of $\mathcal{G}(\mathcal{N},\mathcal{E})$, which can reach all necessary buses with the smallest number of D-FACTS devices. More details are presented in Appendix \ref{sec:deployment}. 

\textbf{\textit{Principle 2:}} From the robust consideration, the following lemma is derived for the attacks targeting on the weakest point(s) for the grid with incomplete MTD configuration.

\begin{lemma}\label{lemma:vul}
Let $\bm{U} = (\bm{U}_{1},\bm{U}_{2,3})$ where $\bm{U}_{2,3}$ is the collection of columns in $\bm{U}_2$ and $\bm{U}_3$. Let $\bm{a}_N = \bm{U}_{1}\bm{c}_{1} + \bm{U}_{2,3}\bm{c}_{2,3}$ with $\bm{c}_{1}\in\mathbb{R}^{k}$ and $\bm{c}_{2,3}\in\mathbb{R}^{r+l}$. The detection rate on $\bm{a}_N$ does not depend on the value of $\bm{c}_{1}$.
\end{lemma}

\begin{proof}
Please refer to Appendix \ref{sec:proof_vul}.
\end{proof}

Although the attackers cannot immediately know the exact $\bm{x}'$ (Assumption 1), Lemma \ref{lemma:vul} suggests that the MTD algorithm should be designed such that the attackers cannot easily attack on $\bm{U}_1$ by chance. Specifically, considering the attack targeting a single state $i$, if $\text{Col}(\bm{J}_N(:,i))\subseteq	\bm{U}_1$, the single-state attack on the bus $i$ can bypass the MTD while any attack involving bus $i$ can be detected ineffectively. To avoid ineffective MTD on 
this attack, 
the following constraint is considered.
\begin{equation}\label{eq:idle_mtd}
    \|\bm{P}_N^i\bm{P}_N'\|_2 \geq \gamma_i, \quad \forall i\in\mathcal{N}^c
\end{equation}
where $\bm{P}_N^i = \left(\bm{J}_N(:,i)^T\bm{J}_N(:,i)\right)^{-1}\bm{J}_N(:,i)\bm{J}_N(:,i)^T$ is the orthogonal projector on $\text{Col}(\bm{J}_N(:,i))$. $\mathcal{N}^c$ represents the index set of buses that are included in at least a loop\footnote{As proved by \cite{zhang2020hiddenness}, if a bus is not included in any loop, attacks on this bus cannot be detected regardless of the MTD strategies.} of $\mathcal{G}$. Since $\|\bm{P}_N^i\bm{P}_N'\|\in[0,1]$ and 1 is achieved when $\text{Col}(\bm{J}_N(:,i))\subseteq\bm{U}_1$, the threshold $\gamma_i$ can be set close but not equal to 1. 

Notice that the constraint in \eqref{eq:idle_mtd} cannot eliminate the weakest point(s) nor improve the worst-case detection rate on $\bm{U}_1$, but it can restrict the attacker's knowledge on the weakest point(s). Rewriting $\lambda$ as $\lambda = \|(\bm{I}-\bm{P}_N')\sum_{i=1}^n\bm{J}_N(:,i)\bm{c}(i)\|_2^2$, constraint \eqref{eq:idle_mtd} ensures that $(\bm{I}-\bm{P}_N')\bm{J}_N(:,i)\bm{c}(i)\neq0$, $\forall i \in \mathcal{N}^c$. To have low MTD detection rate, the attacker has to coordinate the attack strength on at least two buses to have low $\lambda$ which is only possible if $\bm{x}'$ is known. As long as the attacker cannot easily attack $\bm{U}_1$, the probability of having the worst case is low and the MTD strategy is still 
effective from a robust point of view.

\begin{remark}
To fulfill constraint \eqref{eq:idle_mtd}, all buses in $\mathcal{N}^c$ should be incident to at least a branch equipped with D-FACTS devices, which can be achieved by the proposed D-FACTS devices placement algorithm in Appendix \ref{sec:deployment}.
\end{remark}

\textbf{\textit{Principle 3:}} Although the chance of the worst-case attack is minimized by Principle 1-2, it does not necessarily imply a high detection rate when $\bm{a}_N \notin \bm{U}_1$. Similarly to \eqref{eq:robust_opt_2}, the minimum non-zero principal angle $\theta_{k+1}$, which represents the weakest point in subspace $\mathcal{J}_N\setminus(\mathcal{J}_N'\cap\mathcal{J}_N)$ should be maximised by
\begin{equation}\label{eq:incomplete_0}
    \begin{array}{cl}
         \min_{\Delta\bm{x}} & {\cos{\theta_{k+1}}} \\
        \text{s.t.} & \eqref{eq:reactance_limit_in}-\eqref{eq:reactance_limit_out}, \eqref{eq:idle_mtd}
    \end{array}
\end{equation}
where $\cos{\theta_{k+1}}$ is the $(k+1)$th largest singular value. 

To our knowledge, there is no direct method to solve \eqref{eq:incomplete_0} as finding the singular value at a certain position requires solving the SVD of $\bm{P}_N\bm{P}_N'$ and locating the $1$th to $k$th singular vectors. Therefore, we propose an iterative Algorithm \ref{alg:incomplete} to solve \eqref{eq:incomplete_0}. In line 1 of Algorithm \ref{alg:incomplete}, a warm start $\Delta \bm{x}^0$ is firstly found by minimising the Frobenius norm $\|\cdot\|_F$, which is shown to be an upper bound to $\cos{\theta_{k+1}}$.
\begin{equation}\label{eq:incomplete_fro}
    \begin{array}{cl}
         \min_{\Delta\bm{x}} & {\|\bm{P}_N\bm{P}_N'\|_F} \\
        \text{s.t.} & \eqref{eq:reactance_limit_in}-\eqref{eq:reactance_limit_out}, \eqref{eq:idle_mtd}
    \end{array}
\end{equation}

For a given warm-start perturbation $\Delta \bm{x}^0$, the intersection subspace $\bm{U}_1$ can be located by Algorithm \ref{alg:vulnerable_point}. Denoting $\bm{U}_1(\Delta \bm{x}^0)$ as $\bm{U}_1^0$, the t-SVD \eqref{eq:com_svd} can be rewritten as
\begin{equation*}
\begin{array}{rl}
    \bm{P}_N\bm{P}_N' = & \left(\bm{U}_1^0, \bm{U}_{2,3}\right)\begin{pmatrix}
    \bm{I} & \bm{0}\\
    \bm{0} & \Gamma_{2,3}
\end{pmatrix}\begin{pmatrix}
    \bm{V}_{1}^{0T}\\
    \bm{V}_{2,3}^T
\end{pmatrix}  \\
    = & \bm{U}_{1}^0\bm{U}_{1}^{0T} + \bm{U}_{2,3} \Gamma_{2,3} \bm{V}_{2,3}^T
\end{array}
\end{equation*}
where $\bm{I}$ is the identity matrix of dimension $k$; $\Gamma_{2,3} = \text{diag}([\cos(\theta_{k+1}),\cdots,\cos(\theta_{n}]))$ with $\theta_{k+1}\neq0$. Note that $\bm{U}_1^0 = \bm{V}_1^0 = \mathcal{J}_N'\cap\mathcal{J}_N$.

Therefore, the following optimisation problem can be formulated to minimise $\cos{\theta_{k+1}}$: 
\begin{equation}\label{eq:incomplete_1}
    \begin{array}{cl}
         \min_{\Delta\bm{x}} & \|\bm{P}_N\bm{P}_N' - \bm{U}_{1}^0\bm{U}_{1}^{0T}\|_2 \\
        \text{s.t.} & \eqref{eq:reactance_limit_in}-\eqref{eq:reactance_limit_out}, \eqref{eq:idle_mtd}
    \end{array}
\end{equation}

Denoting the optimal value of \eqref{eq:incomplete_1} as $\Delta \bm{x}^1$, a new intersection subspace $\bm{U}_1^1= \bm{U}_1(\Delta\bm{x}^1)$ can be located. As $\Delta\bm{x}^1$ is solved with fixed $\bm{U}_1^0$, $\bm{U}_1^1$ may not be the same as $\bm{U}_1^0$. After finding the new intersection subspace from $\Delta \bm{x}^1$, \eqref{eq:incomplete_1} can be iteratively solved until convergence, as shown by line 3-11 in Algorithm \ref{alg:incomplete}.

\begin{algorithm}[t]
    \footnotesize
    \SetKwInOut{Input}{Input}
    \SetKwInOut{Output}{Output}

    \Input{grid topology $\mathcal{G}(\mathcal{N},\mathcal{E})$, terminating tolerance $tol$, maximum iteration number $max\_ite$}
    \Output{reactance perturbation $\Delta\bm{x}^1$}
    Find the warm start point $\Delta \bm{x}^0$ by solving \eqref{eq:incomplete_fro};
    
    Find the intersection subspace $\bm{U}_{1}^0$ by Algorithm \ref{alg:vulnerable_point};
    
    \tcc{iteration until convergence.}
    
    \While{$step<max\_ite$}
    {
    Find $\Delta\bm{x}^1$ by solving \eqref{eq:incomplete_1};
    
    Find the intersection subspace $\bm{U}_{1}^1$ by Algorithm \ref{alg:vulnerable_point};
    
    \eIf{$\|\bm{U}_1^1-\bm{U}^0_1\|_2 \leq tol$}
    {
    break; \tcc{converged.}
    }
    {
    $\bm{U}_{1}^0:=\bm{U}_{1}^1$;
    }
    }
    \caption{Robust MTD for the Grid with Incomplete Configuration}
    \label{alg:incomplete}
\end{algorithm}

To sum up, Algorithm \ref{alg:incomplete} limits the chance of attacking on $\mathcal{J}_N'\cap\mathcal{J}_N$ (Principal 1-2) and guarantees the worst-case detection rate in $\mathcal{J}_N\setminus(\mathcal{J}_N'\cap\mathcal{J}_N)$ (Principal 3 and \eqref{eq:incomplete_fro}-\eqref{eq:incomplete_1}) for the grid with incomplete configuration.

\subsection{Discussions on Full AC Model Design}

\label{sec:discussion}

In previous sections, we theoretically established the robust MTD algorithm based on the simplified AC model \eqref{eq:linear_equation}-\eqref{eq:linear_attack}. There exists similar concept on the weakest point in the original AC settings \eqref{eq:power_balance}-\eqref{eq:ac_se}. Let $\bm{h}'^{-1}(\cdot)$ represent the result of state estimation in \eqref{eq:ac_se}. The estimated state on attacked measurement is written as $\hat{\bm{\nu}}_a' = \bm{h}'^{-1}(\bm{z}_a')$ and the residual is $\gamma_a'=\|\bm{R}^{-\frac{1}{2}}(\bm{z}_a' - \bm{h}'(\hat{\bm{\nu}}_a'))\|_2^2$. The weakest point can be defined as a unitary attack vector such that $\gamma_a'$ is minimised. However, there are several obstacles to analytically writing its expression. Firstly, recall that $\bm{a} = \bm{h} (\bm{\nu}'+\bm{c}) - \bm{h} (\bm{\nu}')$ which is non-linearly dependent on the post-MTD state $\bm{\nu}'$ and the state attack vector $\bm{c}$. Note that $\bm{\nu}'$ is dependent on $\bm{x}'$ which cannot be determined in advance. Second, $\bm{h}'^{-1}(\cdot)$ requires an iterative update, such as the Gauss-Newton or Quasi-Newton algorithm. Although it is possible to reformulate AC-SE as semi-definite programming \cite{zhu2014power}, it lacks of analytical solution in general. Third, it is difficult to define the concept of angles between subspaces defined by two functions $\bm{h}(\cdot)$ and $\bm{h}'(\cdot)$. Consequently, we theoretically derived the robust algorithm based on the simplified AC model and numerically verify the performance on AC-FDI attacks in simulation. We found out that the MTD designed by the sufficient separation between the subspaces between the real-time Jacobian matrices can provide effective detection in the full AC model.



\section{Simulation}

\label{sec:simulation}

\subsection{Simulation Set-ups}

We test the proposed algorithms on IEEE benchmarks case-6, case-14, and case-57 in MATPOWER \cite{zimmerman2011matpower}. AC-OPF is solved using the Python package PYPOWER 5.1.15. and the nonlinear optimisation problems are solved using the open source library SciPy. More simulation setups are given below.

\subsubsection{Attack Pools and BDD Threshold}
Firstly, we define the attack strength with respect to the noise level as:
\begin{equation}\label{eq:att_rho}
    \rho = \frac{\|\bm{a}\|_2}{\sqrt{\sum_i^m \sigma_i^2}}
\end{equation}

We consider three types of attacks for the simplified AC model. 1). \textbf{\textit{Worst-case attack}} where the attacker attacks on the weakest point $\bm{u}_{k+1}$ of a given MTD strategy according to Algorithm \ref{alg:vulnerable_point}; 2). \textbf{\textit{Single-state attack}} where the attacker only injects on single non-reference phase angle; and 3). \textbf{\textit{Random attack}} where the attack vector $\bm{a}$ is randomly generated as follows. First, the number of attacked state $\|\bm{c}\|_0 = q$ is drawn uniformly from set $\{1,2,\dots,n\}$. $\bm{c}$ is then sampled from multivariate Gaussian distribution with $q$ non-zero entries. Second, the attack vector is found as $\bm{a}=\bm{J}\bm{c}$ and rescaled by different $\rho = 5,7,10,15,20$ according to \eqref{eq:att_rho}. To simplify the analysis, the measurement noise is set as $\sigma_i = 0.01p.u., \forall i$ in all case studies. In this case, to have $\beta$-MTD, the necessary condition is $\rho\geq\sqrt{\lambda_c(\beta)/m}$ according to Proposition \ref{prop:min_att}.

In the original AC model, the measurement consists of $P_i$, $Q_i$, $P_{k:i\to j}$, and $Q_{k:i\to j}$ \eqref{eq:power_balance}, which are nonlinearly dependent on $\bm{\theta}$. Therefore, we randomly sample $\bm{c}$ from uniform distribution and classify $\bm{a} = \bm{h}'(\bm{\nu}'+\bm{c}) - \bm{h}'(\bm{\nu}')$ into one of the ranges $\{[5,7), [7,10), [10,15), [15,20), [20,25), [25,\infty)\}$ by \eqref{eq:att_rho}.

We sample \texttt{no\_load}=50 load conditions on a uniform distribution of the default load profile in MATPOWER \cite{zimmerman2011matpower} for each grid. We then set the D-FACTS devices using different MTD algorithms and simulate the real-time measurements. Under each load condition, we generate \texttt{no\_attack}=200 attack attempts for each of the attack types. The BDD threshold $\tau_\chi(\alpha)$ is determined with $\alpha=5\%$ FPR.

\subsubsection{Metrics and Baselines} 

The key metric to evaluate the MTD detection performance is the true positive rate, also known as the attack detection probability (ADP), which is the ratio between the number of attacks that are detected by the MTD detector and the total number of attacks.

The max-rank MTD algorithm modified from \cite{liu2018reactance, liu2020optimal, zhang2019analysis, zhang2021strategic} is compared as the baseline where reactances are randomly changed with $\mu_{min}\bm{x}_i \leq |\Delta \bm{x}_i|\leq \mu_{max}\bm{x}_i$. Note that each reactance is perturbed by $\mu_{min}>0$ to fulfil the max-rank condition on the composite matrix. For each attempt of attack \texttt{no\_attack}, we simulate \texttt{no\_maxrank} = 20 MTDs of maximum rank to evaluate their average detection performance.

\subsection{Verification of Theoretical Analysis on Simplified AC Model}

\begin{figure}
     \centering
     \begin{subfigure}[b]{0.24\textwidth}
         \centering
         \includegraphics[width=\textwidth]{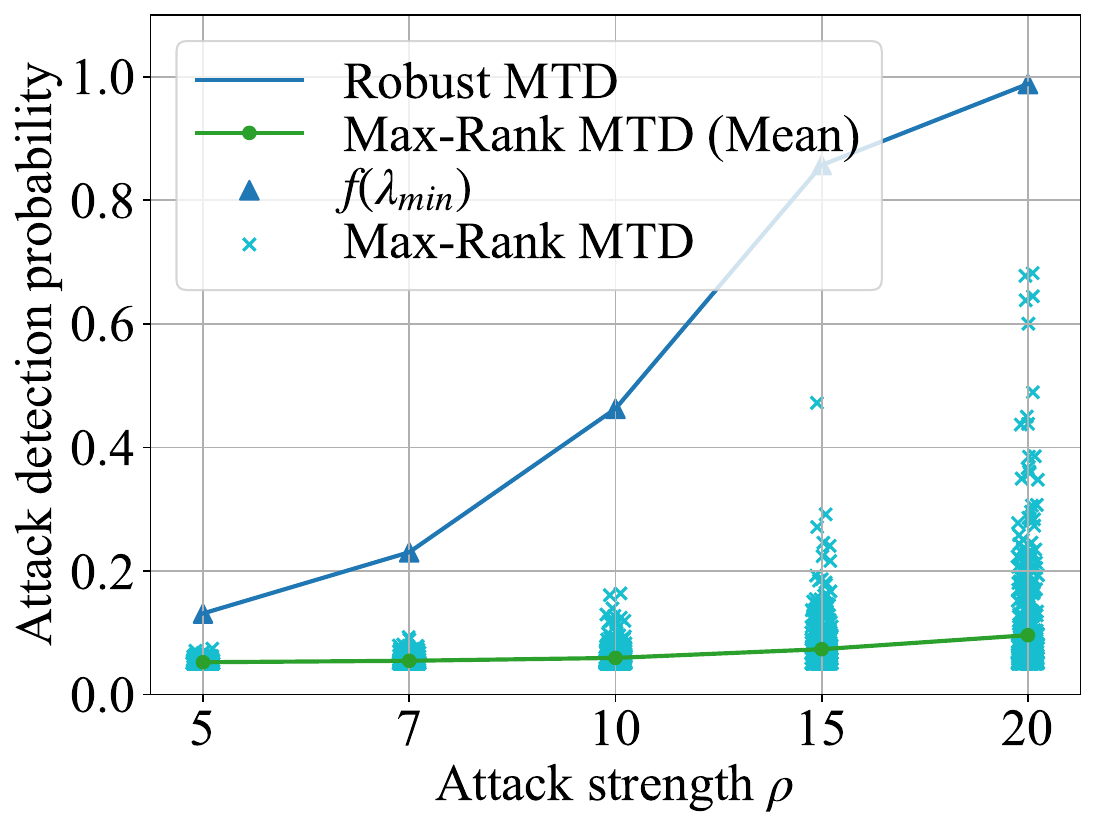}
         \caption{Worst-Case Attack}
     \end{subfigure}
     \hfill
     \begin{subfigure}[b]{0.24\textwidth}
         \centering
         \includegraphics[width=\textwidth]{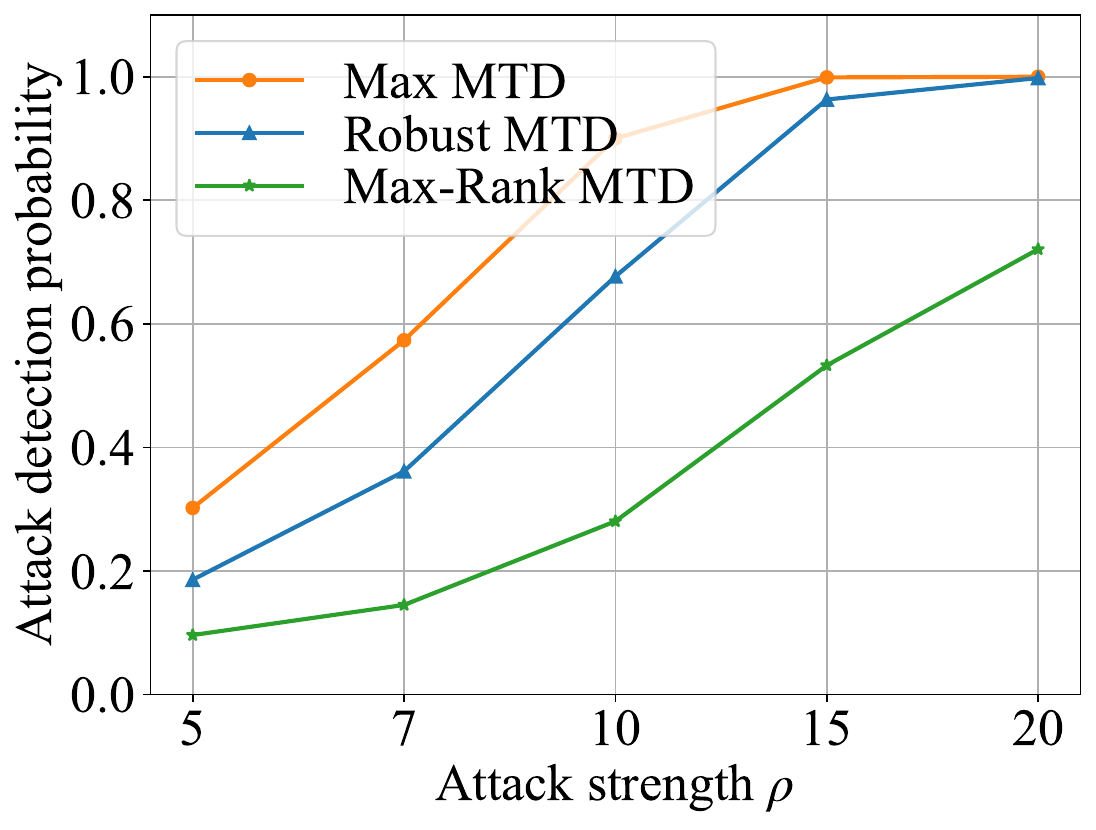}
         \caption{Random Attack}
     \end{subfigure}
        \caption{ADPs on simplified case-6 system.  }
        \label{fig:case6}
\end{figure}

\begin{figure}
     \centering
     \begin{subfigure}[b]{0.24\textwidth}
         \centering
         \includegraphics[width=\textwidth]{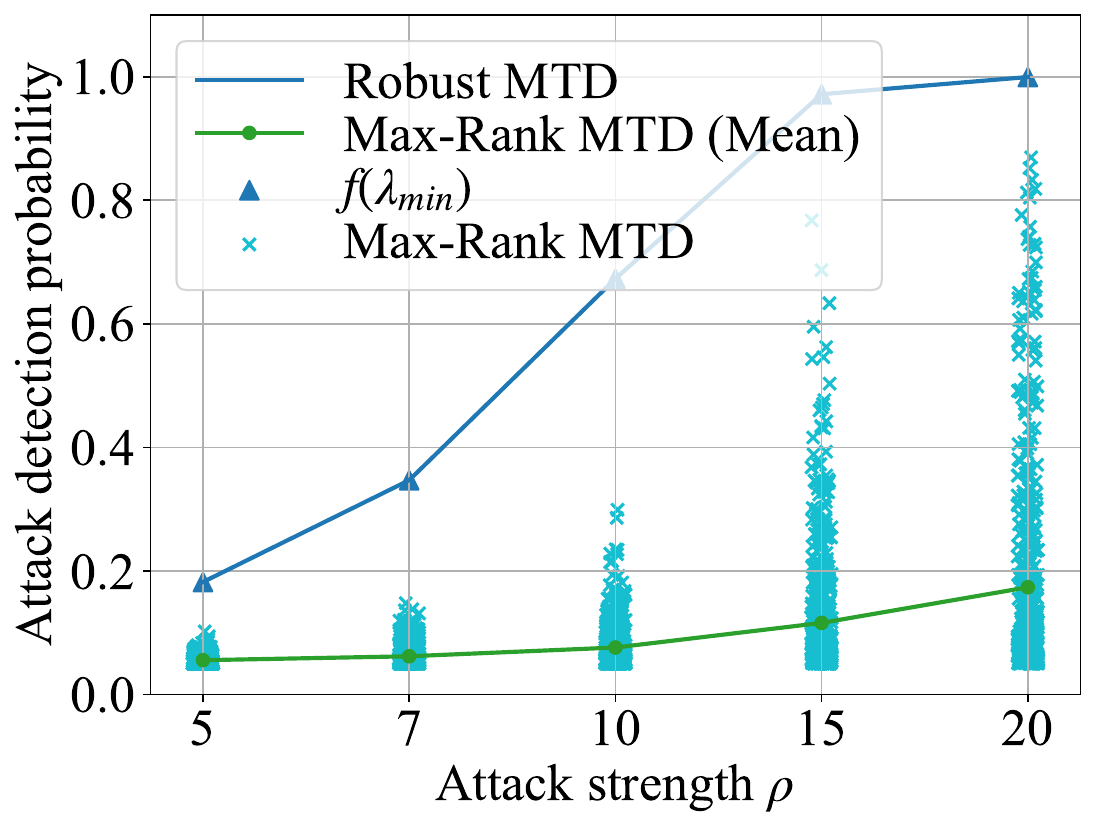}
         \caption{Worst-Case Attack}
     \end{subfigure}
     \hfill
     \begin{subfigure}[b]{0.24\textwidth}
         \centering
         \includegraphics[width=\textwidth]{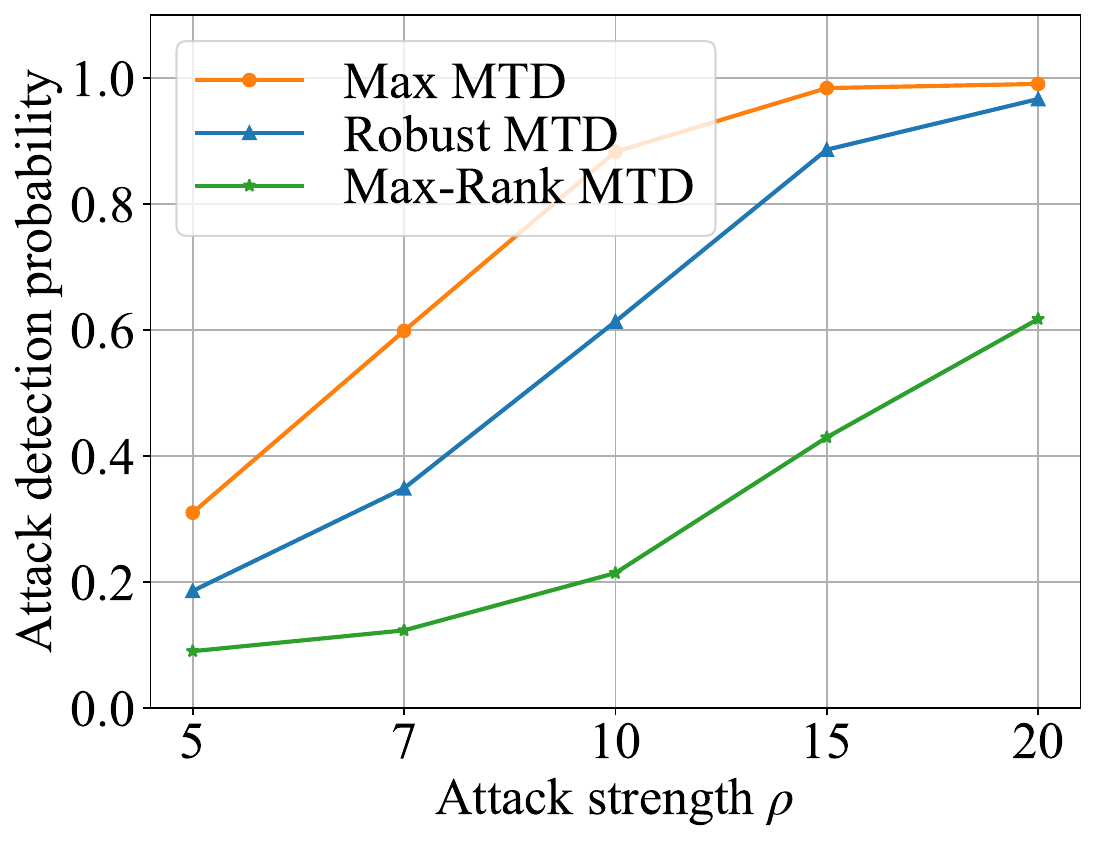}
         \caption{Random Attack}
     \end{subfigure}
        \caption{ADPs on simplified case-14 system.  }
        \label{fig:case14}
\end{figure}

In the first case study, we verify the theoretical analysis of robust MTD algorithms and demonstrate their effectiveness in the simplified AC model \eqref{eq:linear_equation}-\eqref{eq:linear_attack}. 

First, the ADPs of case-6 with complete configuration are illustrated in Fig. \ref{fig:case6} for both worst-case attacks and random attacks. The reactances are changed with
$\tau = 0.2$. Meanwhile, $\mu_{min} = 0.05$ and $\mu_{max} = 0.2$ in the max-rank MTD. In Fig. \ref{fig:case6}(a), the simulation result on the ADPs of robust MTD is the same as the theoretic detection rate $f(\lambda_{min})$ calculated by Proposition \ref{prop:min_eff}, which verifies the theoretic analysis and the design criteria. In addition, the robust MTD algorithm shows much higher ADPs than the max-rank MTD on the worst-case attack. Although the max-rank MTD's performance may approach the robust MTD in some cases, its average ADP is similar to the FPR as the worst-case performance cannot be explicitly considered under the noiseless setting.   

In Fig. \ref{fig:case6}(b), the max MTD is added by solving \eqref{eq:max_mtd} with the assumption that the attack vector $\bm{a}_N$ is known, which represents the performance upper-bound of any MTD design. As shown by Fig. \ref{fig:case6}(b), the robust MTD algorithm, not only guarantees the worst case condition, but also outperforms the max-rank MTD by 10\%-45\% on random attacks with different $\rho$. 
Moreover, the gap between robust MTD and max MTD algorithms is smaller than 25\% and approaches to zero when $\rho\geq15$. However, comparing Fig. \ref{fig:case6}(a) and Fig. \ref{fig:case6}(b), it is worth noting that the major improvement of robust MTD over max-rank MTD still lies in the worst-case attacks.

Fig. \ref{fig:case14} investigates the performance on the case-14 system with incomplete configuration. By Algorithm \ref{alg:vulnerable_point}, the minimum $k$ is equal to 6 and the worst point in $\mathcal{J}_N\setminus(\mathcal{J}_N'\cap\mathcal{J}_N)$ is at $\bm{u}_7$. Assume that all branches are equipped with D-FACTS devices and the maximum perturbation ratio is set as $\tau = 0.2$. Although the detection rates on attacks in $\bm{U}_1$ are equal to $\alpha$ according to Lemma \ref{lemma:vul}, the ADP on $\bm{u}_7$ is nonzero by implementing Algorithm \ref{alg:incomplete} and increases as the strength of the attack increases. Similar to Fig. \ref{fig:case6}(a), although the max-rank MTD algorithm can, by chance, give a high detection rate against the worst-case attack, its average detection rate is extremely low. In Fig. \ref{fig:case14}(b), the gap between the max MTD and the robust MTD is also small (5\%-30\%). The results demonstrate that robust design can also effectively improve the detection performance for the grid with incomplete configuration. 


To further investigate on the weakest points in $\bm{U}_1$, we generate single-bus attack with $\rho = 10$ and record the ADPs in Fig. \ref{fig:single_state_attack} with and without Principle 2 \eqref{eq:idle_mtd}. First, attacks targeting bus-8 can only be detected by 5\%. This is because bus-8 is a degree-one bus which is excluded by any loop. Second, with Principle 2 considered, the robust MTD can give more than 90\% ADPs for all buses. In contrast, there are attacks against certain buses, e.g. bus-7, 10, 11, and 13 can be barely detected without Principle 2. 
Consequently, the simulation result verifies that Principle 2 can sufficiently reduce the chance of attacking on the weakest points. 

\begin{figure}
     \centering
     \includegraphics[width=0.24\textwidth]{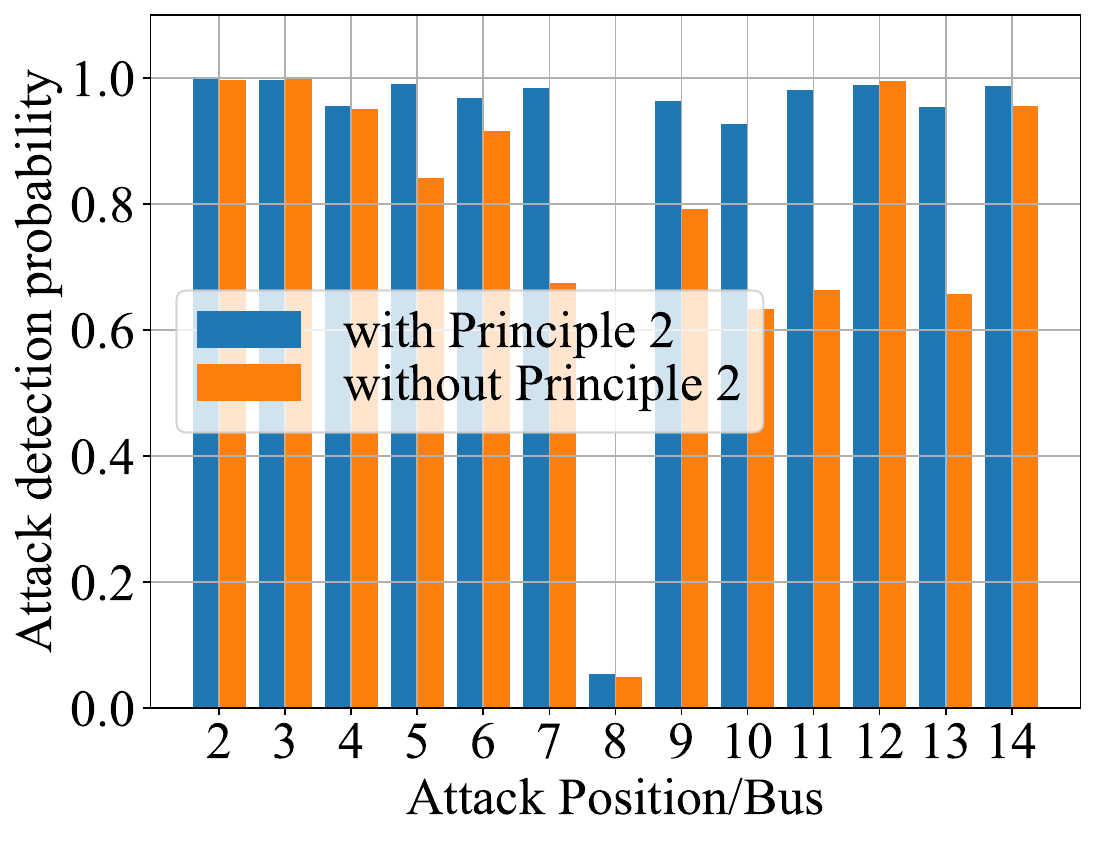}
    \caption{ADPs on single-state attacks of case-14 system.}
    \label{fig:single_state_attack}
\end{figure}

\subsection{Simulation Results on Full AC Model}

In this section, we verify the detection effectiveness of the proposed robust MTD algorithms on FDI attacks under the original AC settings \eqref{eq:power_balance}-\eqref{eq:ac_se}.

\subsubsection{Random Attack}

Random attacks ADPs for the full-AC cases-6, case-14, and case-57 systems are summarised in Table \ref{tab:random_ac}. Similar to studies on simplified AC models, the proposed robust algorithms can improve ADPs by 10\%-40\% compared with the max-rank algorithm. In particular, for cases with attack strength below 20, robust MTD can almost double the ADPs of max-rank MTD for all three systems. Therefore, the robust MTD designed by the principal angles between the subspaces of pre- and post- MTD Jacobian matrices are still effective on defending AC-FDI attacks. In addition, the
attacks with larger attack strength are more likely to be detected while the detection probability for different systems under the same attack strength is slightly different due to their different load levels, parameters (e.g. the reactance to resistance ratios), and topologies. For instance, case-57 system is harder to detect as the ADPs in both max-rank and robust MTDs are lower than the case-6 and case-14 systems. 

To confirm detection performance, the residual distributions for the three systems are summarised in Fig. \ref{fig:ac_result} where kernel density estimation is used to smooth the histograms. 
The result implies that the proposed algorithms can generalise well to the AC-FDI attacks by sufficiently shifting the distribution positively, which is shown to be a key property on effective MTD with the measurement noise considered in Fig. \ref{fig:residual_illustration}. For each sub-figure, the max-rank MTD performs worse than the robust MTD on average as well. 

\begin{table}
\footnotesize
\caption{Average ADPs on random AC-FDI attacks. Max-Rk represents the max-rank MTD, and Robust represents the robust MTD.}
  \begin{tabular}{c|cccccc}
    \hline
    \multirow{2}{*}{$\rho$} &
      \multicolumn{2}{c}{\textbf{case-6}} &
      \multicolumn{2}{c}{\textbf{case-14}} &
      \multicolumn{2}{c}{\textbf{case-57}} \\
    & Max-Rk & Robust & Max-Rk & Robust & Max-Rk & Robust \\
    \hline\hline
    $[5,7)$   & 7.1\%  & 13.7\% &  8.6\% & 18.1\% & 10.3\% & 30.3\% \\
    \hline
    $[7,10)$  & 12.6\% & 33.2\% & 14.4\% & 41.2\% & 15.2\% & 39.2\% \\
    \hline
    $[10,15)$ & 25.1\% & 67.3\% & 27.5\% & 63.1\% & 23.7\% & 55.9\% \\
    \hline
    $[15,20)$ & 44.5\% & 92.4\% & 43.4\% & 87.5\% & 36.0\% & 69.1\%\\
    \hline
    $[20,25)$ & 60.2\% & 98.2\% & 60.6\% & 94.5\% & 50.6\% & 81.6\%\\
    \hline
  \end{tabular}
  \label{tab:random_ac}
\end{table}

\begin{figure}
     \centering
     \begin{subfigure}[b]{0.24\textwidth}
         \centering
         \includegraphics[width=\textwidth]{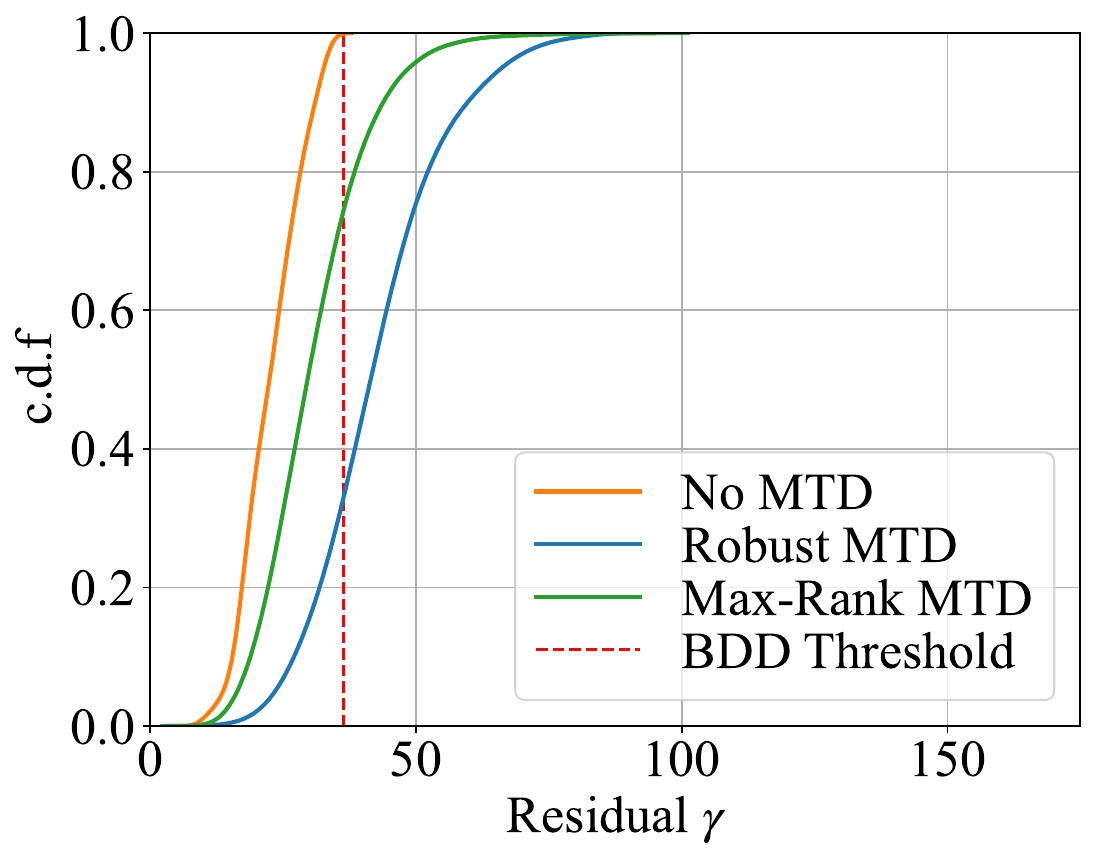}
     \end{subfigure}
     \hfill
     \begin{subfigure}[b]{0.24\textwidth}
         \centering
         \includegraphics[width=\textwidth]{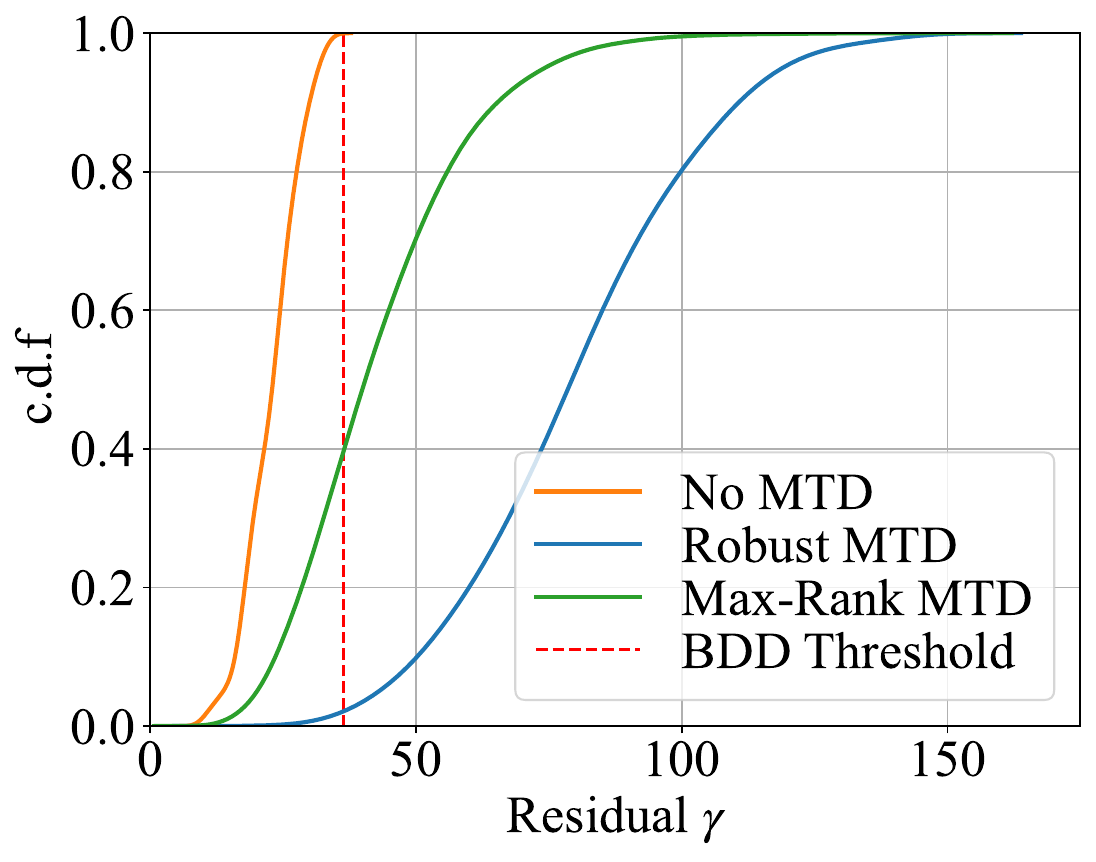}
     \end{subfigure}
     \hfill
     \begin{subfigure}[b]{0.24\textwidth}
         \centering
         \includegraphics[width=\textwidth]{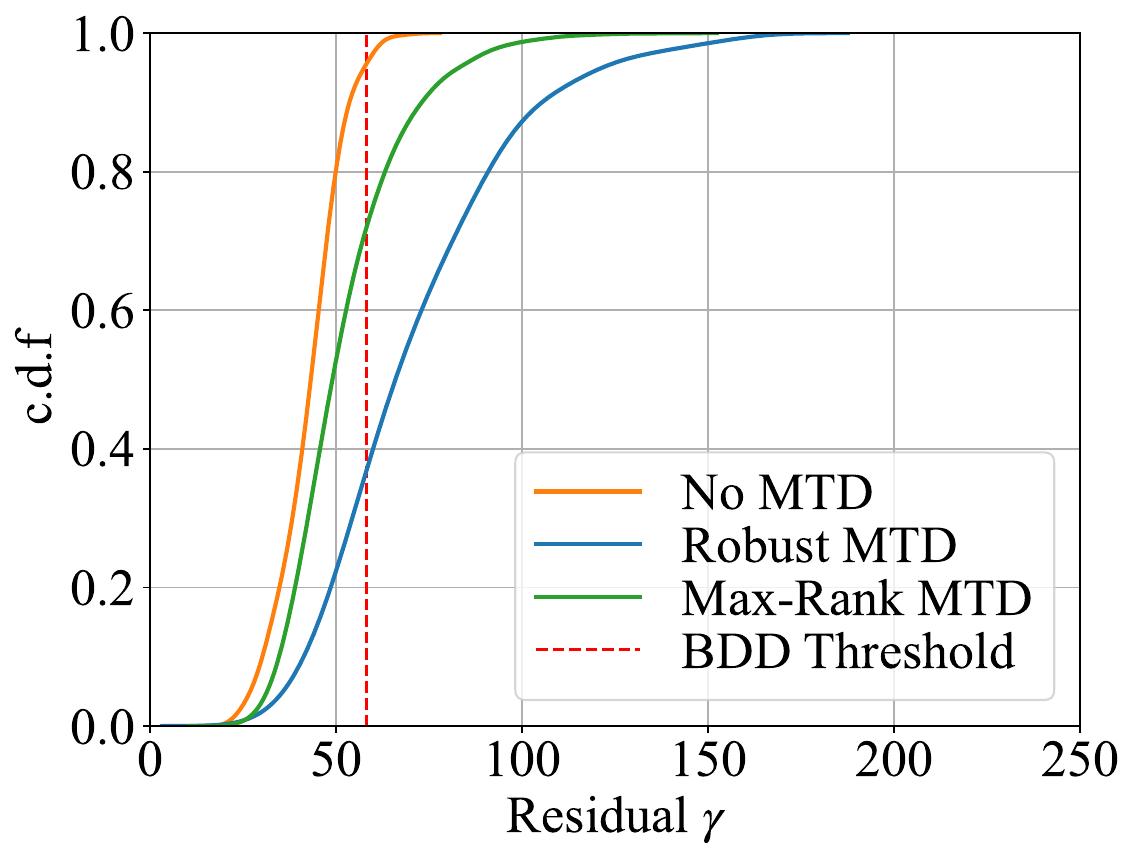}
     \end{subfigure}
     \hfill
     \begin{subfigure}[b]{0.24\textwidth}
         \centering
         \includegraphics[width=\textwidth]{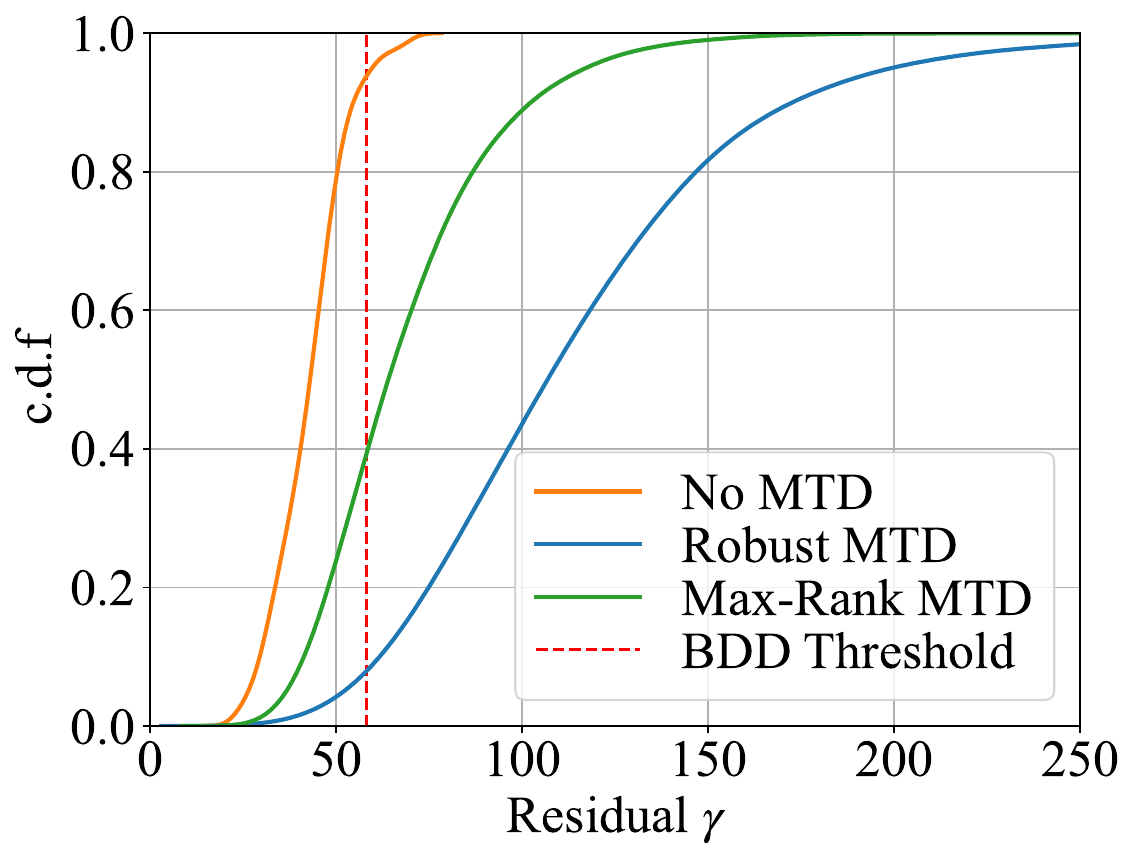}
     \end{subfigure}
     \hfill
     \begin{subfigure}[b]{0.24\textwidth}
         \centering
         \includegraphics[width=\textwidth]{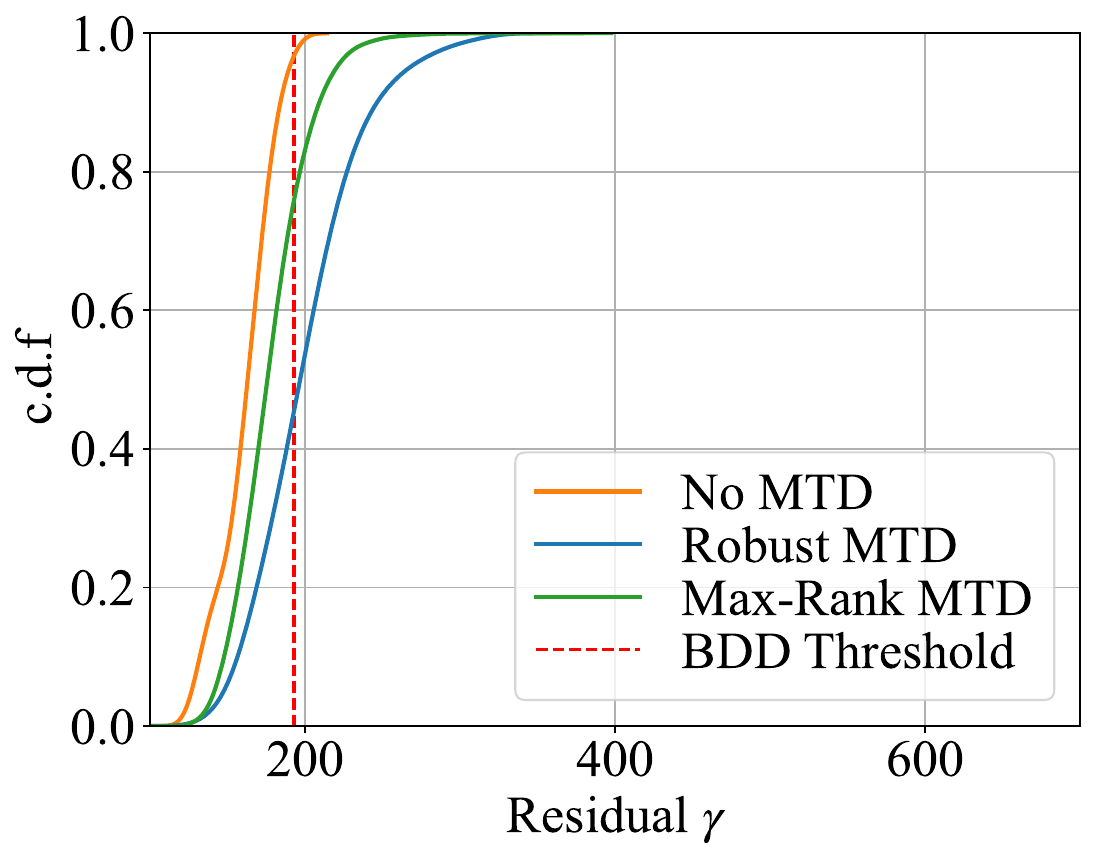}
     \end{subfigure}
     \hfill
     \begin{subfigure}[b]{0.24\textwidth}
         \centering
         \includegraphics[width=\textwidth]{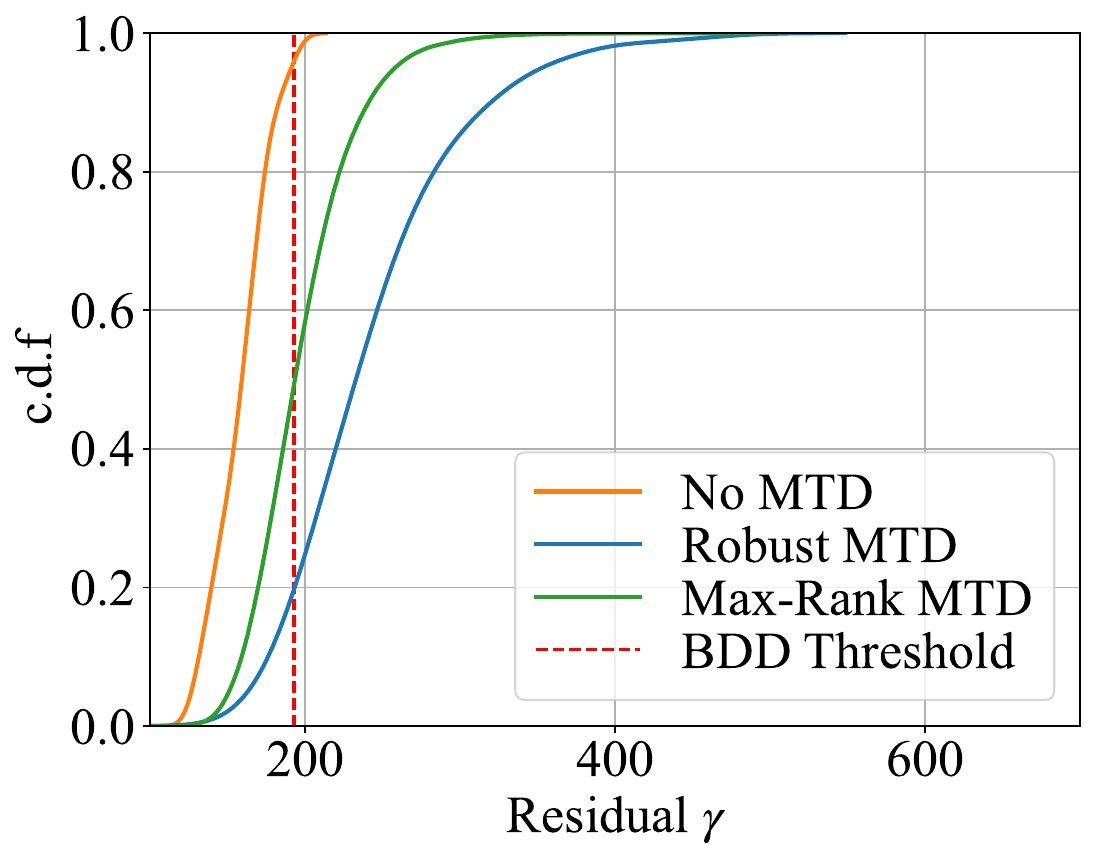}
     \end{subfigure}
        \caption{Residual distributions of AC-FDI attacks. The first row: case-6 system; the second row: case-14 system; the third row: case-57 system; the first column: attacks in range $[10,15)$; the second column: attacks in range $[20,25)$.}
        \label{fig:ac_result}
\end{figure}

\subsubsection{Impact of Different Placements and Perturbation Ratios of D-FACTs Devices}

Fig. \ref{fig:case14_random_compare} records the simulation results on AC random attacks under two different D-FACTS devices placements and four different perturbation ratio limits. In detail, `all' represents perturbing all branches, whereas `part' represents perturbing on branch- 2, 3, 4, 12, 15, 18, and 20, which is the outcome of the `D-FACTS Devices Placement Algorithm' in Appendix \ref{sec:deployment}. The simulation result shows that $k=6$ is achieved and all buses are covered except bus 8 in `part' placement. As the maximum perturbation ratio is reported as 50\% in literature \cite{lakshminarayana2021cost}, $\tau$ is set as 0.2, 0.3, 0.4, and 0.5. As a result, the grey curve in Fig. \ref{fig:case14_random_compare} is simulated in the same settings as the robust MTD in Table \ref{tab:random_ac}. When the number of D-FACTS devices is limited, although the minimum $k$ is still met by Principle 1, the detection rate is significantly reduced. To attain a higher detection rate, the perturbation limit should be further increased. 
Notably, the dependence of ADP on different D-FACTS device placements and perturbation ratios can only be found when the sensor noise is considered.

\begin{figure}
    \centering
    \includegraphics[width = 0.24\textwidth]{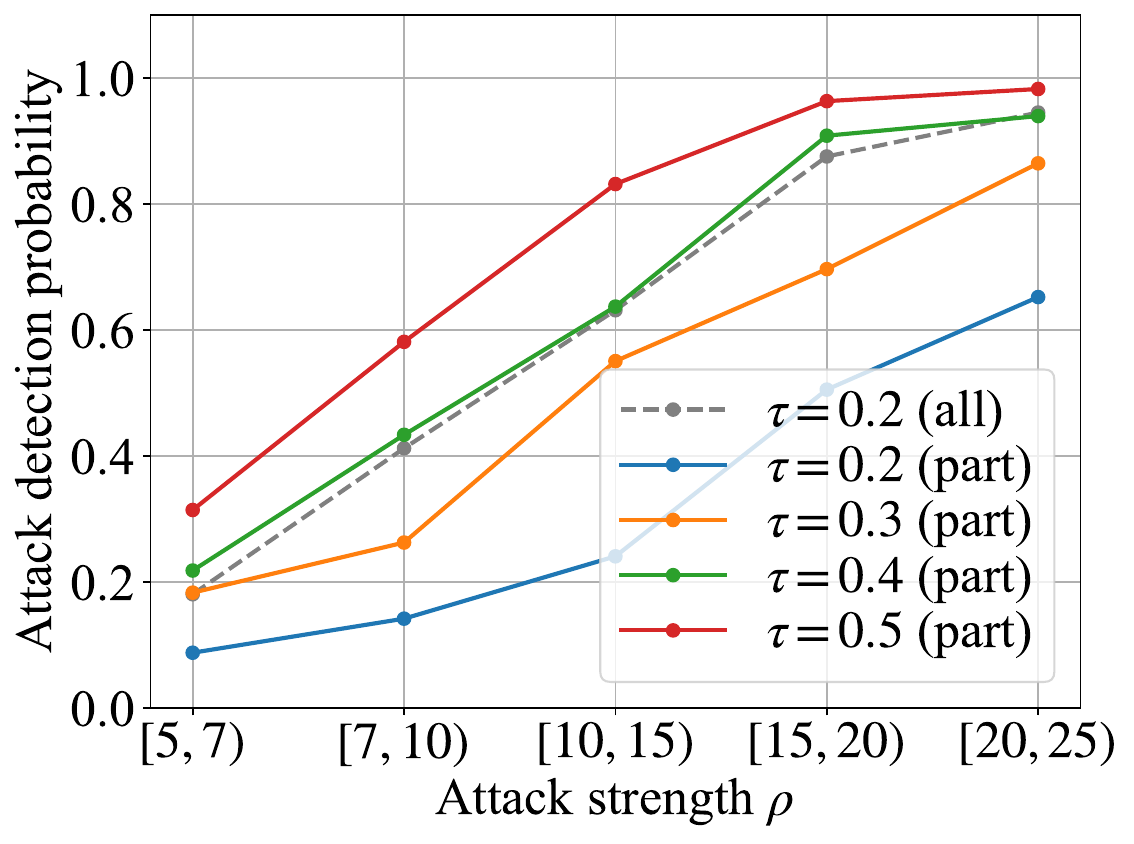}
    \caption{ADPs under different placements and perturbation ratios of D-FACTS devices.}
    \label{fig:case14_random_compare}
\end{figure}

\subsubsection{Computational Time}

The computational time of the proposed algorithms are summarised in Table \ref{tab:time}. We test the proposed algorithm on the MacBook Pro with Apple M1 Pro chip and 32 GB memory. For each system and algorithm, the computational times under all load conditions are recorded and averaged. The multi-run strategy is also applied to approach the global optimum of the nonlinear optimisation problem which is also included in Table \ref{tab:time}. Although the computation time depends on the system scales, number of D-FACTS devices, and algorithms, they are acceptable for real-time applications. In practise, as attackers spend time collecting new measurements and learning new parameters \cite{lakshminarayana2021cost}, the system operator can solve robust MTD algorithms with a period much longer than the state estimation time, e.g., several hours, or only change the Jacobian matrix $\bm{J}_N$ when the loads are significantly changed. A flat state vector may also be a choice to construct the Jacobian matrix if the loads change slowly. 

\begin{table}[h]
    \centering
    \footnotesize
    \caption{Computational Time (averaged by \texttt{no\_load} runs). }
    \begin{tabular}{ c|ccc } 
        \hline
        \textbf{Case} & \textbf{No. D-FACTS}  & \textbf{Algorithm} & \textbf{Time (s)} \\\hline\hline
        case-6 & 11 & \eqref{eq:complete_mtd} & 0.022 \\\hline
        \multirow{3}{4em}{case-14} & 20 & Algorithm \ref{alg:incomplete} & 1.925 \\ 
        & 20 & Algorithm \ref{alg:incomplete} without \eqref{eq:idle_mtd} & 0.325 \\ 
        & 7 & Algorithm \ref{alg:incomplete} & 0.532 \\\hline
        case-57 & 78 & Algorithm \ref{alg:incomplete} & 9.357 \\
        \hline
        \end{tabular}
    \label{tab:time}
\end{table}

\section{Conclusions}

\label{sec:conclusion}

In this paper, we address the real-time robust implementation of MTD against unknown FDI attacks. Using the concept of angles between subspaces, we theoretically prove that the weakest point for any given MTD strategy corresponds to the smallest principal angle and the worst-case detection rate is proportional to the sine of this angle, with the impact of measurement noise being explicitly considered. These novel findings can help evaluate the effectiveness of any MTD strategy. Moreover, a robust MTD algorithm is proposed by increasing the worst-case detection rate for the grid with complete MTD configuration. We then demonstrate that the weakest point(s) of incomplete MTD always exist and cannot be improved. Therefore, robust MTD is proposed for the grid with incomplete configuration by refraining from the ineffective MTD operation and improving the worst-case detection rate in the detectable subspace. The simulation results on standard IEEE benchmarks verify the effectiveness of real-time detection in AC-FDI attacks, compared with the baseline. 
In the future, we would like to cooperate the proposed robust MTD algorithm with hiddenness being considered. Meanwhile, a constrained optimisation problem can also be derived to minimise the usage of D-FACTS devices.

\appendix

\numberwithin{equation}{section}
\setcounter{equation}{0}

\subsection{Proof of Proposition \ref{prop:com_incom}}
\label{sec:proof_com_incom}
The composite matrix of the original and perturbed Jacobian matrix \eqref{eq:jacobian} is written as:
\begin{equation*}
    \begin{pmatrix}
    \bm{J} & \bm{J}'
    \end{pmatrix} = \bm{V}\begin{pmatrix}
    \bm{B} & -\bm{G} & \bm{B}' & -\bm{G}
    \end{pmatrix}\begin{pmatrix} \bm{A}_{r}^{\cos} & \bm{0} \\ \bm{A}_{r}^{\sin} &  \bm{0} \\ \bm{0} & \bm{A}_r^{\cos} \\ \bm{0} & \bm{A}_r^{\sin} \end{pmatrix}
\end{equation*}

Given the property of the matrix product, the rank of the composite matrix satisfies $\text{rank}((\bm{J}\quad\bm{J}'))\leq \min \{m,m,2n\}$. If $m<2n$, $\text{rank}((\bm{J}\quad\bm{J}')) \leq m < 2n$ no matter how the D-FACTS devices are altered. Therefore, the MTD cannot be complete if $m<2n$.

\subsection{Normalised Measurement Vectors and Matrices}
\label{sec:app_normalize}

We consider measurement noise follows independent Gaussian distribution which is not necessarily isotropic. Let $\bm{z}_N = \bm{R}^{-\frac{1}{2}}\bm{z}$, $\bm{e}_N = \bm{R}^{-\frac{1}{2}}\bm{e}$, and $\bm{J}_N = \bm{R}^{-\frac{1}{2}}\bm{J}$. The measurement equation becomes $\bm{z}_N = \bm{J}_N\bm{\theta} + \bm{e}_N$. $\bm{P}_{\bm{J}}$, which is defined on $\langle\,,\rangle_{\bm{R}^{-\frac{1}{2}}}$, now becomes $\bm{P}_{\bm{J}_N} = \bm{J}_N(\bm{J}_N^T\bm{J}_N)^{-1}\bm{J}_N^T$, defined on $\langle\,,\rangle$. Similarly, $\bm{S}_{\bm{J}_N} = \bm{I} - \bm{P}_{\bm{J}_N}$. It is easy to show that $\bm{R}^{-\frac{1}{2}}\bm{S}_{\bm{J}} = \bm{S}_{\bm{J}_N}\bm{R}^{-\frac{1}{2}}$. As a result, $\bm{r}(\bm{z}_N) = \bm{S}_{\bm{J}_N}\bm{e}_N$ follows (approximately) standard normal distribution $\bm{r}(\bm{z}_N)\sim\mathcal{N}(\bm{0},\bm{I})$. For convenience, we write $\bm{P}_{\bm{J}_N}$ and $\bm{S}_{\bm{J}_N}$ as $\bm{P}_N$ and $\bm{S}_N$ in short.

\subsection{Proof of Proposition \ref{prop:min_att}}\label{sec:app_min_att}

First, a $\beta$-MTD has $\|\bm{S}_N'\bm{a}_N\|_2\geq\sqrt{\lambda_c(\beta)}$. The necessary condition then follows from $\|\bm{S}_N'\bm{a}_N\|_2\leq\|\bm{S}_N\|_2\|\bm{a}_N\|_2 = \|\bm{a}_N\|_2$. 

Moreover, as $\bm{a}_N = \bm{R}^{-\frac{1}{2}}\bm{a}$, it also gives $\|\bm{S}_N'\|_2\|\bm{R}^{-\frac{1}{2}}\|_2\|\bm{a}\|_2 = \|\bm{R}^{-\frac{1}{2}}\|_2\|\bm{a}\|_2\geq\sqrt{\lambda_c(\beta)}$. As $\|\bm{R}^{-\frac{1}{2}}\|_2 = \max\sigma(\bm{R}^{-\frac{1}{2}})=\sigma_{min}^{-1}$, it can be derived that $\|\bm{a}\|_2\geq\sigma_{min}\sqrt{\lambda_c{(\beta})}$. Furthermore, if $\bm{R} = \text{diag}([\sigma,\sigma,\cdots,\sigma])$ is isotropic, it gives $\|\bm{R}^{-\frac{1}{2}}\bm{a}\|_2 = \sigma^{-1}\|\bm{a}\|_2\geq\sqrt{\lambda_c(\beta)}$. Let $\rho = \|\bm{a}\|_2/\sqrt{\sum_{i}^m \sigma_i^2}$. We can result in $\rho\geq\sqrt{\lambda_c(\beta)}/\sqrt{m}$.

\subsection{Proof of Proposition \ref{prop:min_eff}}\label{sec:app_proof_1}
According to Definition \ref{def:vul}, the weakest point $\bm{j}_N^*\in\mathcal{J}_N, \|\bm{j}_N^*\|_2 = 1$ can be derived by
\begin{equation}\label{eq:app_1}
    \begin{array}{rl}
        \bm{j}_N^* = & \arg\min_{\bm{j}_N\in \mathcal{J}_N \atop \| \bm{j}_N\|_2 = 1}\sqrt{\lambda_{eff}} \\
        = & \arg\min_{\bm{j}_N\in \mathcal{J}_N \atop \| \bm{j}_N\|_2 = 1} \frac{\|\bm{j}_N - \bm{P}_N'\bm{j}_N\|_2}{\|\bm{j}_N\|_2} \\
        = & \arg\min_{\bm{j}_N\in \mathcal{J}_N \atop \| \bm{j}_N\|_2 = 1} \sin{\angle\{\bm{j}_N,\bm{P}_N'\bm{j}_N\}}
    \end{array}
\end{equation}

Note that the triangle relationship within the sides $\|\bm{j}_N\|$, $\|\bm{P}_N'\bm{j}_N\|$, and $\|\bm{j}_N-\bm{P}_N'\bm{j}_N\|$ and the ratio in \eqref{eq:app_1} is the sine of the angle between the vectors $\bm{j}_N$ and $\bm{P}_N'\bm{j}_N$. Basing on the definition of principal angle \eqref{eq:minimal_angle}, the sine of the angle is minimized when $\angle\{\bm{j}_N,\bm{P}_N'\bm{j}_N\} = \theta_1$. The minimum principal angle is achieved when $\bm{j}_N$ and $\bm{P}_N'\bm{j}_N$ are reciprocal such that $\bm{j}_N = \bm{u}_1$ and $\bm{P}_N'\bm{j}_N = \bm{P}_N'\bm{u}_1 = \cos{\theta_1}\bm{v}_1$ \cite{galantai2008subspaces, ben2003generalized}. 

Moreover, the worst-case detection rate is achieved when attacking on $\bm{u}_1$ such that
\begin{equation*}
    \lambda_{{min}} = \|a\bm{u}_1 - a\cos{\theta_1}\bm{v}_1\|_2^2 = a^2\sin^2{\theta_1}
\end{equation*}

\subsection{D-FACTS Devices Placement}\label{sec:deployment}

A modified minimum edge covering algorithm is proposed to find the smallest number of D-FACTS devices covering all buses while satisfying the minimum $k$ condition. 
The pseudocode is given by Algorithm \ref{alg:mtd}.
In detail, the inputs to the proposed MTD deployment algorithm are the grid information $\mathcal{G}(\mathcal{N},\mathcal{E})$ and the output is branch set $\mathcal{E}_D$. 
On lines 1-2, $\text{CB}$ represents the function to calculate the set of cycle bases of a given graph. The algorithm \ref{alg:mtd} then removes any buses that are not included by cycle basis (thus not in any loops) and the corresponding branches from the grid $\mathcal{G}$. 
In line 3-4, the minimum edge covering (MEC) problem is solved. Given the power grid topology, MEC firstly runs the maximum (cardinality) matching algorithm to find the maximum branch set whose ending buses are not incident to each other \cite{bondy2008graph}. The maximum matching is found by Edmonds’ BLOSSOM algorithm where the size of the initial empty matching is increased iteratively along the so-called augmenting path spotted by blossom contraction \cite{bondy2008graph}. 
After constructing the maximum matching, a greedy algorithm is performed to add any uncovered buses to the maximum matching set. The resulting set of branches becomes $\mathcal{E}_D$, the minimum edge covering set where each bus is connected to at least one branch. Lines 5-15 guarantee the minimum $k$ requirement where it breaks the edge in any identified cycle bases in $\overline{\mathcal{G}}_2$. At last, line 11-13 is added to avoid adding any new loop in $\overline{\mathcal{G}}_1$.

\begin{algorithm}
    \footnotesize
    \SetKwInOut{Input}{Input}
    \SetKwInOut{Output}{Output}

    \Input{grid topology $\mathcal{G}(\mathcal{N},\mathcal{E})$}
    \Output{branch set with D-FACTS devices $\mathcal{E}_D$}
    $\mathcal{L} = \text{CB}(\mathcal{G})$; \tcc{find the circle basis}
    
    Find buses ${\mathcal{N}_1}$ not in $\mathcal{L}$. Remove $\mathcal{N}_1$ and the incident branches from $\mathcal{G}$. Name the resulting graph as $\overline{\mathcal{G}}(\overline{\mathcal{N}},\overline{\mathcal{E}})$\; 

    $\mathcal{E}_{min} = \text{MEC}(\overline{\mathcal{G}})$, construct $\overline{\mathcal{G}}_1(\overline{\mathcal{N}},\mathcal{E}_{min})$ and $\overline{\mathcal{G}}_2(\overline{\mathcal{N}}, \mathcal{E}_r)$ with $\mathcal{E}_r = \overline{\mathcal{E}} \setminus \mathcal{E}_{min} $;
    
    $\mathcal{L}_2 = \text{CB}(\overline{\mathcal{G}}_2)$  \tcc{loops in non D-FACTs graph}
    \For{loop in $\mathcal{L}_2$}
    {
    \For{e in loop}
    {Construct $\overline{\mathcal{G}}_1(\overline{\mathcal{N}}, \mathcal{E}_{min})$ and $\overline{ \mathcal{G}}_2(\overline{\mathcal{N}},\mathcal{E}_r$) where $\mathcal{E}_{min} \leftarrow \mathcal{E}_{min}+e$ and $\mathcal{E}_{r} \leftarrow \mathcal{E}_r-e$\;
    
    $\mathcal{L}_1 = \text{CB}(\overline{\mathcal{G}}_1)$\; \tcc{loops in D-FACTs graph}
    
    \eIf{$\mathcal{L}_1 = \varnothing$}{break}
    {$\overline{\mathcal{G}}_1(\overline{\mathcal{N}}, \mathcal{E}_{min})$ and $\overline{\mathcal{G}}_2(\overline{\mathcal{N}}, \mathcal{E}_r)$ where $\mathcal{E}_{min} \leftarrow \mathcal{E}_{min}-e$ and $\mathcal{E}_{r} \leftarrow \mathcal{E}_r+e$\;} 
    }
    }
    
    
    

    \caption{D-FACTS Devices Placement Algorithm}
    \label{alg:mtd}
\end{algorithm}

\subsection{Proof of Lemma \ref{lemma:vul}}

\label{sec:proof_vul}

Rewrite the non-centrality parameter as
\begin{equation}
\label{eq:vul}
\begin{array}{rl}
    \sqrt{\lambda} = & \|(\bm{I} - \bm{V}\bm{V}^T)\bm{Uc}\|_2 \\
    = & \|(\bm{U} - \bm{V}\Gamma)\bm{c}\|_2 \\
    = & \|\left((\bm{U}_1,\bm{U}_{23}) - (\bm{V}_1\Gamma_1,\bm{V}_{23}\Gamma_{23})\right)\bm{c}\|_2
\end{array}
\end{equation}

As $\bm{U}_1=\bm{V}_1$ and $\Gamma_1 = \bm{I}$, \eqref{eq:vul} can be reduced to $\sqrt{\lambda} = (\bm{U}_{23} - \bm{V}_{23}\Gamma_{23})\bm{c}_{23}$ which does not depend on $\bm{c}_1$.

\bibliographystyle{IEEEtran}
\bibliography{IEEEabrv,Reference}


\begin{IEEEbiography}[{\includegraphics[width=1in,height=1.25in,clip,keepaspectratio]{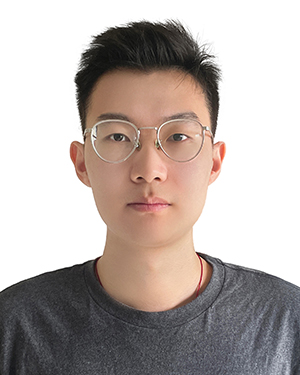}}]{Wangkun Xu} (Student Member, IEEE) received B.Eng. degree in electrical and electronic engineering from Xi-an Jiaotong Liverpool University, China and University of Liverpool, UK, in 2018. He received M.Sc. degree in control systems from Imperial College London, in 2019, where he is currently a Ph.D. student. His research focuses on robust and privacy-preserving machine learnings in power system operation and security.
\end{IEEEbiography}

\begin{IEEEbiography}[{\includegraphics[width=1in,height=1.25in,clip,keepaspectratio]{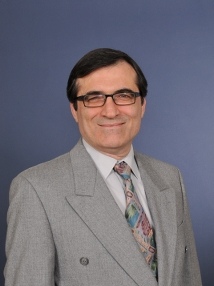}}]{Imad M. Jaimoukha} received the B.Sc. degree in electrical engineering from the University of Southampton, Southampton, U.K., in 1983, and the M.Sc. and Ph.D. degrees in control systems from Imperial College London, London, U.K., in 1986 and 1990, respectively. He was a Research Fellow with the Centre for Process Systems Engineering at ICL from 1990 to 1994. Since 1994, he has been with the Department of Electrical and Electronic Engineering, ICL. His research interests include robust and fault-tolerant control, system approximation, and global optimization.
\end{IEEEbiography}

\begin{IEEEbiography}[{\includegraphics[width=1in,height=1.25in,clip,keepaspectratio]{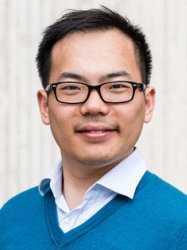}}]{Fei Teng} (Senior Member, IEEE) received the B.Eng. degree in electrical engineering from Beihang University, China, in 2009, and the M.Sc. and Ph.D. degrees in electrical engineering from Imperial College London, U.K., in 2010 and 2015, respectively, where he is currently a Senior Lecturer with the Department of Electrical and Electronic Engineering. His research focuses on the power system operation with high penetration of Inverter-Based Resources (IBRs) and the Cyber-resilient and Privacy-preserving cyber-physical power grid.
\end{IEEEbiography}

\end{document}